\documentclass[11pt]{article}

\usepackage{graphicx, mathtools, amssymb, amsmath, amsfonts, amsthm, dsfont}
\usepackage{array, fancyhdr, xcolor, bm, listings, tikz, soul, titlesec, cancel,caption, subcaption}
\usepackage[toc,page]{appendix}

\usetikzlibrary{patterns,positioning,calc,backgrounds,snakes}

\usepackage[numbers]{natbib}
\usepackage[libertine]{newtxmath}
\usepackage{newtxtext}

\usepackage[margin=1in]{geometry}

\usepackage{hyperref}
\hypersetup{
    colorlinks=true, 
    citecolor=violet,
    linkcolor=blue, 
    urlcolor=red, 
    linktoc=all 
}

\titleclass{\subsubsubsection}{straight}[\subsection]
\newcounter{subsubsubsection}[subsubsection]
\renewcommand\thesubsubsubsection{\thesubsubsection.\arabic{subsubsubsection}}
\titleformat{\subsubsubsection}
  {\normalfont\normalsize\bfseries}{\thesubsubsubsection}{1em}{}
\titlespacing*{\subsubsubsection}
{0pt}{3.25ex plus 1ex minus .2ex}{1.5ex plus .2ex}
\def\toclevel@subsubsubsection{4}
\def\l@subsubsubsection{\@dottedtocline{4}{7em}{4em}}
\theoremstyle{plain}
\newtheorem{theorem}{Theorem}[section]
\newtheorem{lemma}[theorem]{Lemma}
\newtheorem{claim}[theorem]{Claim}
\newtheorem{proposition}[theorem]{Proposition}
\newtheorem{fact}[theorem]{Fact}
\newtheorem{corollary}[theorem]{Corollary}

\newtheorem*{lemma*}{Lemma}
\newtheorem*{claim*}{Claim}
\newtheorem*{proposition*}{Proposition}
\newtheorem*{fact*}{Fact}
\newtheorem*{corollary*}{Corollary}
\newtheorem*{hint*}{Hint}

\theoremstyle{definition}
\newtheorem{definition}[theorem]{Definition}
\newtheorem{remark}[theorem]{Remark}
\newtheorem{notation}[theorem]{Notation}
\newtheorem{example}[theorem]{Example}

\newtheorem{problem}[theorem]{Problem}
\newtheorem{exercise}[theorem]{Exercise}

\newtheorem*{theorem*}{Theorem}
\newtheorem*{definition*}{Definition}
\newtheorem*{remark*}{Remark}
\newtheorem*{notation*}{Notation}
\newtheorem*{example*}{Example}
\newtheorem*{examples*}{Examples}
\newtheorem*{question*}{Question}
\newtheorem*{problem*}{Problem}
\newtheorem*{solution*}{Solution}
\newtheorem*{intuition*}{Intuition}
\newtheorem*{idea*}{Idea}
\newtheorem*{conjecture*}{Conjecture}

\newcommand{\btheorem}{\begin{theorem}}
\newcommand{\etheorem}{\end{theorem}}
\newcommand{\bproblem}{\begin{problem}}
\newcommand{\eproblem}{\end{problem}}
\newcommand{\bfact}{\begin{fact}}
\newcommand{\efact}{\end{fact}}
\newcommand{\bexercise}{\begin{exercise}}
\newcommand{\eexercise}{\end{exercise}}
\newcommand{\bclaim}{\begin{claim}}
\newcommand{\eclaim}{\end{claim}}
\newcommand{\bcorollary}{\begin{corollary}}
\newcommand{\ecorollary}{\end{corollary}}
\newcommand{\bnotation}{\begin{notation}}
\newcommand{\enotation}{\end{notation}}
\newcommand{\bremark}{\begin{remark}}
\newcommand{\eremark}{\end{remark}}
\newcommand{\blemma}{\begin{lemma}}
\newcommand{\elemma}{\end{lemma}}
\newcommand{\bexample}{\begin{example}}
\newcommand{\eexample}{\end{example}}
\newcommand{\bdefinition}{\begin{definition}}
\newcommand{\edefinition}{\end{definition}}
\newcommand{\bproof}{\begin{proof}}
\newcommand{\eproof}{\end{proof}}

\newcommand\red[1]{{\color{red} #1}}

\newcommand\whitespace{{\color{white}.}}

\newcommand{\bitem}{\whitespace\begin{itemize}} 
\newcommand{\eitem}{\end{itemize}}
\newcommand{\be}{\whitespace\begin{enumerate}} 
\newcommand{\ee}{\end{enumerate}}
\newcommand{\bi}{\whitespace\begin{itemize}} 
\newcommand{\ei}{\end{itemize}}

\newcommand{\textand}{\text{and }}

\def\bal#1\eal{\begin{align*}#1\end{align*}}
\def\beq#1\eeq{\begin{equation}#1\end{equation}}
\def\bitem#1\eitem{\begin{itemize*}#1\end{itemize*}}
\def\bcomment#1\ecomment{}
\newcommand{\tikzcenter}[1]{\begin{center}\begin{tikzpicture}#1\end{tikzpicture}\end{center}}

\newcommand{\ignore}[1]{}

\newcommand\tsup{\textsuperscript}

\newcommand{\lp}{\left(}
\newcommand{\rp}{\right)}
\newcommand{\pa}[1]{\lp#1\rp}
\newcommand{\lb}{\left[}
\newcommand{\rb}{\right]}
\newcommand{\ba}[1]{\lb#1\rb}
\newcommand{\lbr}{\left\{}
\newcommand{\rbr}{\right\}}
\newcommand{\br}[1]{\lbr#1\rbr}

\newcommand{\lf}{\lfloor}
\newcommand{\rf}{\rfloor}

\newcommand{\floor}[1]{\lf {#1} \rf}

\newcommand{\lc} {\lceil}
\newcommand{\rc} {\rceil}

\newcommand{\ceil}[1]{\lc {#1}\rc}

\newcommand{\ab}[1]{\langle#1\rangle}

\providecommand{\abs}[1]{\left\vert#1\right\vert}

\providecommand{\abs}[1]{\left\vert#1\right\vert}

\newcommand\ii{\item}

\newcommand\sar{^*}

\newcommand\nonempty{\neq\emptyset}

\newcommand{\defeq}{\vcentcolon=}

\newcommand{\id}[1]{^{(#1)}}

\newcommand\eps{\epsilon}

\newcommand\half {\frac{1}{2}}

\newcommand\fourth {\frac{1}{4}}

\renewcommand{\Pr}{\mathop{\bf Pr\/}}
\newcommand{\E}{\mathop{\mathbb{E}\/}}

\newcommand\Y{\rY}

\newcommand\NN{\mathbb{N}}

\newcommand\calD{\mathcal{D}}

\newcommand\calI{\mathcal{I}}

\newcommand\calL{\mathcal{L}}
\newcommand\calM{\mathcal{M}}
\newcommand\calN{\mathcal{N}}

\newcommand\calR{\mathcal{R}}

\newcommand\calX{\mathcal{X}}
\newcommand\calY{\mathcal{Y}}

\newcommand\one{\mathds{1}}

\DeclareMathOperator*\poly{poly}

\DeclareMathOperator*\suchthat{s.t.}

\DeclareMathOperator*\Dec{Dec}

\DeclareMathOperator*\OnAdv{OnAdv}

\DeclareMathOperator*\LCS{LCS}

\DeclareMathOperator*\Geometric{Geometric}

\let\next\relax
\let\prev\relax
\DeclareMathOperator*\next{next}
\DeclareMathOperator*\prev{prev}
\newcommand\outdeg{\textnormal{outdeg}}
\newcommand\indeg{\textnormal{indeg}}
\newcommand{\enc}{\mathsf{Enc}}
\newcommand{\dec}{\mathsf{Dec}}

\newcommand{\ptabove}[2]{node[circle,fill,inner sep=1pt,label=above:{#2}](#1){}}
\newcommand{\ptbelow}[2]{node[circle,fill,inner sep=1pt,label=below:{#2}](#1){}}

\begin{document}
\title{Coding against deletions in oblivious and online models\thanks{Research supported in part by NSF grant CCF-1422045.}}
\author{Venkatesan Guruswami\thanks{Email: {\tt venkatg@cs.cmu.edu}} \and Ray Li\thanks{Email: {\tt ryli@andrew.cmu.edu}}}
\date{Carnegie Mellon University \\ Pittsburgh, PA 15213}

\maketitle
\thispagestyle{empty}

\begin{abstract}
We consider binary error correcting codes when errors are deletions.
A basic challenge concerning deletion codes is determining $p_0^{(adv)}$, the \emph{zero-rate threshold of adversarial deletions}, defined to be the supremum of all $p$ for which there exists a code family with rate bounded away from 0 capable of correcting a fraction $p$ of adversarial deletions. 
A recent construction of deletion-correcting codes~\cite{BukhGH17} shows that $p_0^{(adv)} \ge \sqrt{2}-1$, and the trivial upper bound, $p_0^{(adv)}\le\frac{1}{2}$, is the best known. Perhaps surprisingly, we do not know whether or not $p_0^{(adv)} = 1/2$.

  \smallskip
In this work, to gain further insight into deletion codes, we explore
two related error models: oblivious deletions and online deletions,
which are in between random and adversarial deletions in power. In the
oblivious model, the channel can inflict an \emph{arbitrary} pattern
of $pn$ deletions, picked without knowledge of the
codeword. We prove the existence of binary codes of positive rate that
can correct any fraction $p < 1$ of oblivious deletions, establishing
that the associated zero-rate threshold $p_0^{(obliv)}$ equals $1$.
  
\smallskip
For online deletions, where the channel decides whether to delete bit
$x_i$ based only on knowledge of bits $x_1x_2\dots x_i$, define the \emph{deterministic zero-rate threshold for online deletions} $p_0^{(on,d)}$ to be the supremum of $p$ for which there exist deterministic codes against an online channel causing $pn$ deletions with low \emph{average} probability of error. 
That is, the probability that a randomly chosen codeword is decoded incorrectly is small.  
We prove $p_0^{(adv)}=\frac{1}{2}$ if and only if $p_0^{(on,d)}=\frac{1}{2}$.
\end{abstract}

\newpage
\tableofcontents
\thispagestyle{empty}

\newpage
\setcounter{page}{1}

\section{Introduction}

Tremendous progress has been made over the decades, including in
recent years, on the problem of designing error-correcting codes that
can recover from bit (or symbol) errors and erasures, for both
probabilistic and adversarial noise models. The problem of correcting
closely related insertion and deletion errors has also been studied
since the work of Levenshtein in the 60s~\cite{Levenshtein66}, but
has not seen commensurate progress. The difficulty is that in the
presence of deletions, the location of the missing symbol is \emph{not} known to the decoder.
As a result, the sender and receiver lose valuable
synchronization information on the location of various symbols, which renders many of the standard techniques for constructing either
inapplicable or at least very tricky and messy to apply. 
To quote from Mitzenmacher's survey~\cite{Mitzenmacher09}: ``[C]hannels with synchronization errors, including both insertions and deletions as well as more general timing errors, are simply not adequately understood by current theory. 
Given the near-complete knowledge we have [for] channels with erasures and errors . . . our lack of understanding about channels with synchronization errors is truly remarkable."

In particular, for both the random and adversarial noise models, there
are big gaps in our knowledge of the power and limitations of codes
against insertions and deletions.  For concreteness, let us focus our
discussion on the model when there are only deletions, which is
already complicated enough, and captures most of the chief
difficulties associated with coding against synchronization
errors. Note that over large alphabets, in particular those which can
grow with the code length, one can add the coordinate position as a
header to the codeword symbol, and reduce the deletion model to the
simpler erasure model, at the expense of a negligible decrease in rate
(due to the addition of the header). However, for fixed alphabets,
deletions seem much harder to cope with than erasures. We focus on the
most interesting case of binary codes in this work.

In order to set the context and motivation for this work, let us
review the situation for adversarial and random deletions in turn,
before turning to the contributions of this paper.  In each setting we
want a code $C \subseteq \{0,1\}^n$ consisting of binary codewords of
length $n$ that has good rate (where rate $\calR$ means $|C| = 2^{\calR n}$,
and we would like $\calR$ to be as large as possible, and certainly
bounded away from $0$ as $n$ grows).  We also have a deletion fraction
$p\in(0,1)$, roughly denoting, as a fraction of $n$, the number of
bits the adversary can delete.  In adversarial deletions, the
adversary is allowed to delete up to $pn$ bits with full knowledge of the code and the transmitted codeword.  A code $C$ is decodable against $pn$
adversarial deletions if and only if, for any two distinct codewords $c$ and $c'$
in $C$, it is impossible to apply $pn$ (possibly different) deletions
to $c$ and $c'$ and obtain the same result. 
This is easily seen to be equivalent to the condition that the longest common subsequence between any two codewords of $C$ has less than $(1-p)n$ bits. 
By a lemma originally due to Levenshtein~\cite{Levenshtein66}, this condition also ensures that $C$ is capable of correcting any combination of adversarial insertions and deletions totaling $pn$ in number. 
We state this lemma formally as it will be used later.
  \begin{lemma}
    Let $\LCS(C)$ be the maximum length of the longest common subsequence between two distinct codewords in $C$.
    The following are equivalent:
      (1) $\LCS(C)\le n-t-1$,
			(2) $C$ decodable under up to $t$ deletions,
      (3) $C$ decodable under up to $t$ insertions,
      and (4) $C$ decodable under up to $t$ insertions and deletions.
  \end{lemma}

When $p \ge \half$, the adversary can delete $n/2$ bits that includes
either all the $0$'s or all the $1$'s in the codeword, resulting in
just two possible received sequences. Therefore, it is impossible to
correct an adversarial deletion fraction of $\half$ with rate bounded
away from $0$. 
Rather remarkably, we do not know if this trivial limit can be approached.
Denote the \textit{zero-rate threshold}, $p_0\id{adv}$, to be the supremum of all $p$ for which there exist families of codes $C \subseteq \{0,1\}^n$ of size $2^{\Omega_p(n)}$ decodable against $pn$ adversarial deletions. 
In other words, $p_0\id{adv}$ is the threshold at which the capacity of adversarial deletions is zero. 
An outstanding question in the area is whether $p_0\id{adv} = \half$.
Or, is $p_0\id{adv}<\half$, so that, for $p_0\id{adv}<p<\half$, any code $C \subseteq\{0,1\}^n$ that is decodable against $pn$ deletions must have size at most $2^{o(n)}$? 
This was explicitly raised as a key open problem in \cite{Levenshtein02}. 
Upper bounds on the asymptotic rate function in terms of the deletion fraction were obtained in \cite{KulkarniK13}, improving in some respects Levenshtein's bounds~\cite{Levenshtein02}. 
New upper bounds on code size for a fixed number of deletions that improve over \cite{Levenshtein66} were obtained in \cite{CullinaK14}.

Turning to constructions of binary deletion codes, it was shown in~\cite{KashMTU11} that randomly chosen codes of small enough
rate $\calR > 0$ can correctly decode against $pn$ adversarial deletions
when $p\le 0.17$. Even non-constructively, this remained the best
achievability result (in terms of correctable deletion fraction) until
recently. Bukh, Guruswami, and H{\aa}stad \cite{BukhG16, BukhGH17} improved this
and showed that there are binary codes of rate bounded away from $0$
decodable against $pn$ adversarial deletions for any $p<\sqrt2-1$;
furthermore, they also gave an efficient construction of such codes
along with an efficient deletion correcting algorithm. Closing the gap
between $\sqrt2-1$ and $\half$ for the maximum fraction of correctable
deletions, and determining whether or not $p_0\id{adv}=\half$,  remain fascinating open problems.

The other heavily studied model of deletion errors has been \emph{random} deletions. Here the basic noise model is the \textit{binary deletion channel}, which deletes each bit independently with probability $p$.
For random deletions, the zero-rate threshold is known to be 1.
In fact, Drinea and Mitzenmacher \cite{DrineaM06, DrineaM07} proved that the capacity of random deletion codes is at least $(1-p)/9$, which is within a constant factor of the upper bound of $1-p$, the capacity of the simpler binary erasure channel (where the location of the missing bits \emph{are} known at the decoder).
Recently, explicit codes of rate $c(1-p)$ for universal constant $c>0$ efficiently coding against random deletions were constructed in \cite{GuruswamiL17}.
Rahmati and Duman \cite{RahmatiD13} prove that the capacity is at most $0.4143(1-p)$ for $p\ge 0.65$.
For large deletion probabilities, this is, to our knowledge, the best known upper bound. 
For small deletion probabilities, Diggavi and Grossglauser \cite{DiggaviG01} establish that the capacity of the deletion channel for $p\le \half$ is at least $1-h(p)$.  
Kalai et al. \cite{KalaiMS10} proved this lower bound is tight as $p\to0$, and Kanoria and Montanari \cite{KanoriaM13} determined a series expansion that can be used to determine the capacity exactly for small deletion probabilities. 
\vspace{-1ex}
\subsection{Our contributions}
\vspace{-1ex}
With the above context, we now turn to our results in this work. 
Our first result concerns a natural model that bridges between the adversarial and random deletion models, namely \emph{oblivious deletions}. 
Here we assume that an arbitrary subset of $pn$ locations of the codeword can be deleted, but these positions must be picked without knowledge of the codeword. 
The oblivious model is well-motivated in settings where the noise may be mercurial and caused by hard to model physical phenomena, but not by an adversary.

If the code is deterministic, tackling oblivious deletions is clearly
equivalent to recovering from worst-case deletions. We therefore allow the
encoder to be randomized, and require that for every message $m$ and
every deletion pattern $\tau$, most encodings of $m$ can be decoded
from the deletion pattern $\tau$. The randomness used at the encoding
is private to the encoder and is not needed at the decoder, which we
require to be deterministic.  Note that the oblivious model
generalizes random deletions, as deleting each bit independently with
probability $p$ is oblivious to the actual codeword, and with high
probability one has $\approx pn$ deletions.  Of course, any code
which is decodable against $pn$ adversarial deletions is decodable
also against $pn$ oblivious deletions, even without any randomization
in the encoding. Perhaps surprisingly, we prove that in the oblivious
model, the limit of $p \le \half$ does not apply, and in fact one can
correct a deletion fraction $p$ approaching $1$. This generalizes the
result that one can correct a fraction $p \to 1$ of random deletions.

\begin{theorem}[Oblivious]
\label{thm:main-intro}
  \label{thm:oblivious}
  For every $p<1$ there exists a family of constant rate \emph{stochastic} codes that corrects against $pn$ oblivious deletions with high probability. 
  More precisely, for every $p < 1$, there exists $\calR > 0$ and a code family with a randomized encoder $\enc: \{0,1\}^{\calR n} \to \{0,1\}^n$ and (deterministic) decoder $\dec: \{0,1\}^{(1-p)n} \to \{0,1\}^{\calR n}\cup\{\perp\}$ such that for all deletion patterns $\tau$ with $pn$ deletions and all messages $m \in \{0,1\}^{\calR n}$,
  \[ \Pr [ \dec(\tau(\enc(m))) \neq  m ] \le o_n(1) \ , \]
where the probability is over the randomness of the encoder (which is private to the encoder and not known to the decoder).\footnote{The notation 
$\tau(y)$ for $y \in \{0,1\}^n$ denotes the bit string obtained by applying deletion pattern $\tau$ to $y$.}
\end{theorem}
The above result is implied by \emph{deterministic} codes $C$ decodable from arbitrary $pn$ deletions under \emph{average-error} criterion; i.e., there is a decoding function $\dec: \{0,1\}^{(1-p)n} \to C$ such that for every deletion pattern $\tau$ with $pn$ deletions, for \emph{most} codewords $c\in C$, $\dec(\tau(c)) = c$. 
We note that Theorem~\ref{thm:oblivious} is an \emph{existential} result, and the codes it guarantees are not explicitly specified. 
The decoder looks for a codeword which contains the received bit string as a subsequence, and outputting it if there is a unique such codeword.

We also prove an impossibility result for online deletions.
An online channel OnAdv consists of $n$ functions $\{\OnAdv_i:i\in[n]\}$ such that $\OnAdv_i:\calX^i\times\calY^{i-1}\to \calY$, where $X=\{0,1\}$ and $Y=\{\ab{0},\ab{1},\ab{}\}$ is a set of strings and $\ab{}$ denotes the empty string, satisfies $\OnAdv(x_1,\dots,x_i, y_1,\dots,y_{i-1}) \in\{\ab{x_i},\ab{}\}$.
The resulting string received by the output is the concatenation of the outputs of $\OnAdv_1,\dots,\OnAdv_n$.
More concisely, an online channel chooses whether to delete the $i$th bit $c_i$ of a codeword $c$ based on only the first $i$ bits $c_1,\dots,c_i$.
Note that all online adversaries are valid adversaries in the omniscient adversary case, so all codes correcting against $pn$ adversarial deletions can also correct against $pn$ online deletions.
Notice that an online channel can delete all the 0s or 1s in the string, so as in the case of adversarial deletions, it is impossible to communicate with positive rate when the deletion fraction is at least $\half$.
However, as the omniscient adversary has more information than the online channel, it could be the case that one can code against a fraction of deletions approaching $\half$ for the online channel, whereas this is not possible against an omniscient adversary.
Our next result rules out this possibility for deterministic codes against the online channel --- if the zero-rate threshold $p_0\id{adv}$ for adversarial deletions is bounded away from $\half$, then so is the deterministic zero-rate threshold for online deletions. 
\begin{theorem}[Online]
\label{thm:online-intro}
  For any $p>1/(3-2p_0\id{adv})$ and $\calR>0$, there exists a deterministic online channel $\OnAdv$ such that, for sufficiently large $n$, for any code $C$ of block length $n$ and rate $\calR > 0$, and any decoder $\Dec: \{0,1\}^{(1-p)n} \to C$, we have $\Pr_{c\in C}[\Dec(\OnAdv(c))\neq c] \ge \eta$ for some absolute constant $\eta > 0$.
  \label{thm:online}
\end{theorem}
In contrapositive form, the above states that in order to construct codes against adversarial deletions with error fraction approaching $\half$, it suffices to construct deterministic codes decodable in the average case against a fraction of \textit{online} deletions approaching $\half$.
We note our result does not rule out the possibility that there exist \emph{stochastic} codes against a fraction of online deletions approaching $\half$ in the case that $p_0\id{adv}<\half$.
In \S\ref{sec:6-2}, we discuss the subtleties of extending our negative result to rule out stochastic codes against a fraction of online deletions approaching $\half$.

Note that since online deletions contain, as a special case, oblivious deletions, if we insist on deterministic codes where every codeword is correctly decoded, then the question is just as hard as adversarial deletions. Thus, it is important to allow some error, in the form of randomized encoding (i.e. stochastic codes) and/or some small fraction of codewords to be decoded incorrectly (i.e., an average-error criterion). This is why, for our result to be nontrivial, the online strategy of Theorem~\ref{thm:online-intro} ensures that a constant fraction, as opposed to a constant number, of codewords are miscommunicated.

\subsection{Other related work on deletion codes and oblivious/online noise models}
\label{sec:1-2}

Our focus in this work is on binary codes. Recently, there has been much progress in the situation for correcting adversarial deletions over large (but fixed size) alphabets. Guruswami and Wang \cite{GuruswamiW17} proved that for $\eps>0$ there exist polynomial time encodable and decodable codes with alphabet size $\poly(1/\eps)$ correcting $1-\eps$ fraction of deletions with rate $\Omega(\eps^2)$, and codes with alphabet size $\poly(1/\eps)$ correcting $\eps$ fraction of deletions with rate $1-\tilde O(\sqrt{\eps})$.
Guruswami and Li \cite{GuruswamiL16} extended these results for high noise and high rate regimes to correcting insertions and deletions.
Recently, in a beautiful work, Haeupler and Shahrasbi \cite{HaeuplerS17} construct, for any $\eps,\delta>0$, codes with alphabet size $\poly(1/\eps)$ and rate $1-\delta-\eps$ correcting a $\delta$ fraction of insertions and deletions.

The model of oblivious errors (such as bit flips) has has been studied
in the information-theory literature as a particular case of
arbitrarily varying channels with state constraints~\cite{LapidothN98} (see
the related work section of \cite{GuruswamiS16} for more background on this
connection). In particular, for the case of bit flips, the capacity
against the model of $pn$ oblivious bit flips (for $p\le \half$) equals
$1-h(p)$, matching the Shannon capacity of the binary symmetric
channel that flips each bit independently with probability $p$. (This
special case was re-proved in \cite{Langberg08} by a different
simpler random coding argument compared to the original
works~\cite{CsiszarN88,CsiszarN89}.)  
Similarly, the capacity against the model of $pn$ oblivious erasures is $1-p$, matching the Shannon capacity of the binary erasure channel.
Explicit codes of rate approaching $1-h(p)$
to correct $pn$ oblivious bit flips (in the sense of
Theorem~\ref{thm:main-intro}, with randomized encoding) were given in
\cite{GuruswamiS16}. This work also considered computationally bounded
noise models, such as channels with bounded memory or with small
circuits, and gave optimal rate codes for \emph{list} decoding against
those models. These models are more general than oblivious errors, but
still not as pessimistic as adversarial noise.

Notice that in the case of both erasures and errors, the capacity in the oblivious and
random models were the same. It is not clear if this is the case for
deletions. While we show that the zero-rate threshold of oblivious deletions equals the zero-rate threshold of random deletions, the rate of the codes we guarantee in
Theorem~\ref{thm:main-intro} for $p \to 1$ is significantly worse than the
$\Omega(1-p)$ lower bound known for random deletions.

Turning to online models, Chen, Jaggi, and Langberg~\cite{ChenJL15}, in an impressive work, determined the exact capacities of the online erasure channel and the online bit-flip channel. 
A recent work studies a seemingly slight (but in fact fundamental) restriction of the online model where the channel's decision regarding the $i$'th bit depends only on the first $(i-1)$ bits and is independent of the current bit~\cite{DeyJLS16}. 
They proved that in the setting of erasures, the capacity in this restricted online model increases to match the capacity $1-p$ of the binary erasure channel, as opposed to $1-2p$ in the true online model.

To the best of our knowledge, ours is the first work to address both oblivious and online deletions.
We feel that given the large gaps in our understanding of coding against adversarial deletions, and the potential of taking advantage of less stringent models of deletion exhibited in this work, further study of these models seems timely and important. In particular, for online deletions as well, the zero-rate threshold is between $\sqrt2-1$ and
$\half$. It is an interesting challenge to ascertain if one can
take advantage of the online restriction, and push some of the ideas
in \cite{BukhGH17} and this work, to enable decoding a fraction of
online deletions approaching $\half$.

\medskip \noindent \textbf{Organization of this paper.}
In \S\ref{sec:2} we outline our oblivious deletion construction and sketch the proof of why it is correct.
In \S\ref{sec:3} we outline our proof that the online zero-rate threshold for deterministic code, $p_0\id{on,d}$, is $\half$ if and only if the adversarial zero-rate threshold, $p_0\id{adv}$, is $\half$.
In \S\ref{sec:4} we introduce definitions and notation for the remainder of the paper. 
In \S\ref{sec:5} we state and formally prove our construction on codes in the oblivious model with $p$ deletion fraction where $p\in(0,1)$.
In \S\ref{sec:6} we prove that the deterministic zero-rate threshold of the online deletion channel $p_0\id{on,d}$ is also bounded away from $\half$ if $p_0\id{adv}  < \half$.

\section{Outline of oblivious deletion codes}
\label{sec:2}
Instead of proving Theorem~\ref{thm:oblivious} directly, we prove a related theorem for decoding in the \emph{average case}, defined below.
\begin{definition}
  \label{def:average-case}
  We say a (deterministic) binary code $C$ of block length $N$ \textit{decodes $pN$ oblivious deletions in the average case} if for every deletion pattern $\tau$ deleting $pN$ bits, we have
  \begin{equation}
    \abs{\{x\in C: \exists y\in C\,\suchthat\,x\neq y\,\textand\tau(x)\le y\}} \ \le \ o_N(1) \cdot |C| \ .
  \end{equation}
\end{definition}
\begin{theorem}[Correcting oblivious deletions for average message]
  \label{thm:average-case}
  Let $p\in(0,1)$.
  There exists a constant $\calR > 0$ such that there exist infinitely many $N$ for which there is a rate $\calR$ code $C\subseteq\{0,1\}^N$ that decodes $pN$ oblivious deletions in average case.
\end{theorem}
It is standard to show that for oblivious deletions, a randomized code
with low probability of incorrect decoding for \emph{every} message (as promised in Theorem~\ref{thm:oblivious}), and a deterministic code with low probability of incorrect decoding for a random message, are equivalent.  In particular Theorem~\ref{thm:average-case} implies
Theorem~\ref{thm:oblivious}.  For completeness we provide a proof of
this implication in Appendix~\ref{app:1}.
Roughly, if we have a length $N$ code $C$ with rate $\calR$ that
decodes against $pN$ deletions in the average case, then we can group
the codewords into sets of size $2^{0.01\calR n}$.  Then we associate
every message with a set of codewords and encode the message by
randomly choosing a codeword from its set.  Our decoding function
simply takes a received word $s$ and looks for a codeword $c$ such
that $s\sqsubseteq c$, i.e. $c$ is a superstring of $s$.  If there is
exactly one such $c$, output the associated codeword, otherwise output
$\perp$.  Using the fact that $C$ decodes $pN$ deletions in the
average case, we can show that for all deletion patterns $\tau$, only
a few codewords do not decode correctly via unique decoding in our new
stochastic code.

We now outline our construction for Theorem~\ref{thm:average-case}. 
A first, naive attempt at this problem would be to choose a random subset of $2^{\calR N}$ codewords in $\{0,1\}^N$.
This technique, however, does not work in the same way it does for oblivious bit-flips.
See Appendix~\ref{app:2} for a discussion on the difficulties faced by this approach. Instead, our approach to prove Theorem~\ref{thm:average-case} uses the so-called ``clean construction''  from \cite{BukhGH17} with appropriately selected parameters.
The idea is to choose a concatenated code such that the inner code widely varies in the number of runs between codewords.
Specifically, we choose sufficiently large constants $R$ and $K$ and set our inner code to have length $L = 2R^K$.
For $i=1,2,\dots,K$, set our inner codewords to be
\begin{align}
  g_i = \left(0^{R^{i-1}}1^{R^{i-1}}\right)^{L/(2R^{i-1})}
\end{align}
where $0^{k}$ and $1^{k}$ denote strings of $k$ 0s and 1s, respectively.
In this way, the number of runs between any two codewords differs by a factor of at least $R$.
Our outer code is a subset of $[K]^n$, so that the total code length is $N = nL$, and we concatenate the code via a function $\psi:[K]^*\to\{0,1\}^*$ that replaces a symbol $i\in[K]$ with the string $g_i\in\{0,1\}^L$.
The outer code is chosen via a random process, detailed in the following paragraphs; the process throws out a small subset of ``bad'' elements of $[K]^n$ and chooses a constant rate code by including each remaining element independently with some small fixed probability.
Our decoding function is \textit{unique decoding}.
That is, given a received word $s$, we find a codeword $c$ such that $s\sqsubseteq c$. \footnote{The relation $w \sqsubseteq w'$ means that $w$ is a (not-necessarily consecutive) subsequence of $w'$.}
If such a codeword is unique, our decoder returns that output.
Otherwise, the decoder returns $\perp$.

The following example illustrates why the varying run length is powerful even for correct more than $0.5N$ deletions: Suppose $n=100$, $p=0.9$, $R\gg 20$, our received word is $s=g_1^{10}$ (that is, 10 copies of $g_1$ concatenated together) and our code contains the codeword $c=g_2^{100}$.
Then $s$ is a subsequence of $c$ if and only if we can identify each of the 10 $g_1$s with non-overlapping bits of $c$.
However, since $g_1$ contains over $20$ times as many runs as $g_2$, each $g_1$ must be identified with a subsequence of $c$ spanning at least 20 inner codewords, i.e. copies of $g_2$.
This means the subsequence $s$ roughly must span at least $200$ inner codewords, but $c$ only has 100 inner codewords, contradiction.
While this imbalanced run-count behavior is a key to our argument, it is worth noting that the behavior is asymmetric.
In particular, while it takes $R$ copies of $g_2$ to produce $g_1$ as a subsequence, we only need two copies of $g_1$ to produce $g_2$ as a subsequence.

To analyze this code, we leverage the run-count behavior of the inner codewords.
This contrasts with the adversarial setting, where the run-count property of the same code is featured less centrally in the proof of correctness \cite{BukhGH17}.
We show that for any deletion pattern $\tau$ with up to $pN$ deletions, two random codewords $X$ and $Y$ are ``confusable'' with exponentially small probability in $N$.
To be precise, we have for all $\tau$,
\begin{equation}
  \Pr_{X,Y\sim U([K]^n)}[\tau(\psi(X))\sqsubseteq \psi(Y)] \ < \ 2^{-\Omega(n)}. \footnote{The notation $U([K]^n)$ denotes the uniform distribution on $[K]^n$.}
  \label{eq:intro-1}
\end{equation}

The idea for this proof can be illustrated in the case that $\tau$ deletes only entire inner codewords.
If $\tau$ deletes $pn$ of the $n$ inner codewords and does not touch the remaining codewords, then $\tau(\psi(X))$ has the same distribution over length $(1-p)N$ binary strings as $\psi(X')$ where $X'\sim U([K]^{(1-p)n})$. 
Thus we would like to show
\begin{equation}
  \Pr_{\substack{X'\sim U([K]^{(1-p)n})\\Y\sim U([K]^n)}} [\psi(X')\sqsubseteq \psi(Y)] \ < \ 2^{-\Omega(n)}.
  \label{eq:intro-2}
\end{equation}
Consider trying to find $\psi(X')$ as a substring of $\psi(Y)$ where $X'=X_1'X_2'\dots X_{(1-p)n}'$
This is possible if and only if we match the bits of $X'$ to bits of $Y$ greedily.
However, each inner codeword $g_{X_i'}$ spans a large number ($R$) of inner codewords of $Y$ unless the greedy matching encounters at least one $Y_j$ such that $Y_j\le X_i'$, i.e. a higher frequency inner codeword (if $Y_j<X_i'$, we may need to encounter two such higher frequency inner codewords, but two is at least one).
Thus, if $X_i'=1$ for some $i$, then the number of inner codewords spanned by $g_{X_i'}$ is approximately distributed as $\Geometric(1/K)$.
In general, conditioned on $X_i'=k$, the number of inner codewords spanned by $g_{X_i'}$ is approximately distributed as $\Geometric(k/K)$.
Thus, the number of inner codewords of $Y$ spanned by a single inner codeword $X_i'$ is
\begin{equation}
  \frac{1}{K}  \cdot \Theta\left( \frac{K}{1} \right)
  + \frac{1}{K}\cdot \Theta\left( \frac{K}{2} \right)
  + \cdots
  + \frac{1}{K}\cdot \Theta\left( \frac{K}{K} \right)
  \ = \ \Theta(\log K)
  \label{eq:intro-3}
\end{equation}
If we choose $K$ so that $\Theta(\log K) > \frac{2}{1-p}$ then the expected number of inner codewords of $Y$ spanned by $X'$ is more than $\Theta(\log K)\cdot (1-p)n > 2n$, so concentration bounds tell us that the probability that $\psi(X')\sqsubseteq\psi(Y)$ is exponentially small in $n$.

Note that there is a slight caveat to the above argument because the numbers of inner codewords in $\psi(Y)$ spanned by $\psi(X_i')$ are not independent across all $i$.
For example, if $\psi(Y)$ begins with $g_2g_1$ and $\psi(X')$ begins with $g_1g_1$, then the second inner codeword of $\psi(Y)$ has a few ``leftover bits'' that easily match with $\psi(X')$'s second inner codeword.
However, we can adjust for this independence by relaxing our analysis by a constant factor.

The above addresses the case when the deletion pattern, for each inner codeword, either deletes the codeword entirely or does not modify it at all.
We now show how to extend this to general deletion patterns, starting with the case of $p<\half$.

Note that the above argument only depended on the inner codewords having widely varying runs.
By a simple counting argument, we can verify that at least a $(0.5-p)$ fraction of codewords have at most $(p + (0.5-p)/2)L=(0.5+p)L/2$ deletions.
Since we are in the setting where $p<\half$, applying $(0.5+p)L/2$ deletions to an inner codeword with $r$ runs cannot make the number of runs less than $(0.5-p)r$, as it can delete at most $(0.5+p)/2$ fraction of the runs, and deleting each run reduces the number of runs by at most two.
If we choose $R\gg1/(0.5-p)^2$, then we can guarantee that even if we have a generic deletion pattern, we have a constant fraction $(0.5-p)$ of positions for which the run-count properties of all inner codewords in those positions are preserved up to a factor of $\sqrt{R}$.
Thus, as the number of runs between any two inner codewords differs by a factor of at least $R$, even after these corruptions the ratio between the number of runs of two ``preserved'' inner codewords is still at least $\sqrt{R}$.
Using the same argument as above and now requiring $\Theta(\log K) > \frac{2}{0.5-p}$, we can conclude that even for general deletion patterns that the probability that two random candidate codewords are confused is exponentially small.

As a technical note, working with deletion patterns directly is messy as they encode a large amount of information, much of which we do not need in our analysis.
Furthermore, the caveat mentioned in the clean deletion pattern case regarding the independence of number inner codewords spanned by some $\psi(X_i')$ becomes more severe for general deletion patterns. 
This happens because in general deletion patterns, especially later for $p>\half$, the inner codewords of $\psi(X)$ that have preserved their run-count property might nonetheless be shorter, so many (in particular, $\Theta(1/(1-p))$) inner codewords could match to single inner codewords of $\psi(Y)$.
To alleviate this complexity in the analysis, we introduce a technical notion called a \textit{matching} intended to approximate the subsequence relation.
This notion allows us to capture only the run-count behavior of the deletion patterns with respect to the inner codewords while also accounting for the lack-of-independence caveat.
For a deletion pattern $\tau$, let $\sigma$ be the associated deletion pattern such that $\sigma(X)$ to removes all outer codeword symbols except the ones in whose position $\tau$ preserves the run-count property of all the inner codewords (these exist because we are still in the case $p<1/2$).
In our proof, we argue that if $\tau(\psi(X))$ is a subsequence of $\psi(Y)$, then $\sigma(X)$ has a matching in $Y$, and that the probability that $\sigma(X)$ has a matching in $Y$ for two codewords $X$ and $Y$ is exponentially small.

To extend the argument for generic deletion patterns from $p<1/2$ to $p<1$, we must use a ``local list decoding'' idea.
Note that when the number of deletion patterns exceeds $1/2$, for every codeword there exists deletion patterns that destroy all or almost all of the information in the codeword, e.g. the deletion pattern that deletes all the 1s (or 0s) of the codeword.
For this reason, codes cannot correct against more than $1/2$ fraction of adversarial deletions.
However, one can show that this does not happen too frequently allowing us to correct oblivious and average case deletions.
In contrast to the $p<1/2$ case where we found a small, constant fraction of inner codeword positions in which the deletion pattern of the inner codeword preserved the run-count property for \textit{all} inner codewords, we can now find a small, constant fraction of inner codeword positions in which the deletion pattern of the inner codeword preserves the run-count property for \textit{all but a few} inner codewords.
For example, even if an inner code deletion pattern deletes every other bit and thus deletes all the information of $g_1$, the number of runs of $g_2,g_3,\dots,g_K$ are still preserved.
We call this idea ``local list decoding'' because while we cannot decode our constant fraction of inner codewords uniquely, we can still pin down the inner codewords to a few possibilities. 
By extending our definition of matching to account for a few inner codewords potentially losing their run-count behavior, we can prove, just as for $p<1/2$, that $\tau(\psi(X))\sqsubseteq \psi(Y)$ implies $\sigma(X)$ has a matching in $Y$, and $\sigma(X)$ having a matching in $Y$ happens with small probability.

At this point of the proof, we have combinatorially established everything we need to prove that our code is decodable in the average case (and thus against oblivious deletions).
That is, we have shown that a random candidate codeword in $\psi([K]^n)$ has an exponentially small probability of being confused with another random candidate codeword.
Given that codewords have an exponentially small probability of being confusable with each other, it is natural to consider choosing a code by randomly selecting a subset of $[K]^n$.
Using this construction, we might try using concentration bounds to show that, for any deletion pattern $\tau$, the probability that we have more than $\eps|C|$ codewords (for $\eps=o(N)$) that are confusable with some other codeword is at most $2^{-\omega(N)}$, and we can union bound over the at-most-$2^{N}$ choices of $\tau$.
This however does not work directly as the decodability for a given deletion pattern $\tau$ depends on the decodability of other deletion patterns. 
For example, if $p>\half$ and we happen to choose $c=\psi(11\dots1)=0101\dots01$ as a codeword, then for any deletion pattern $\tau$ with $pN$ deletions, $c'\sqsubseteq c$ for \textit{all} candidate codewords $c'$.
From this example alone, the probability of many codewords confusable with $c$ is at least $K^{-n}$ and there are many more examples of such \textit{easily disguised} candidate codewords.
Fortunately, we can prove that the number of easily disguised candidate codewords is small.
In particular, we show that the majority of elements of $\psi([K]^n)$ are not easily disguised in \emph{all} deletion patterns $\tau$.
This intuitively makes sense because, as we have shown, in any deletion pattern $\tau$, on average, words are disguised as an exponentially small fraction of codewords, and because the easily disguised words tend to be easily disguised in every deletion pattern $\tau$.
For example, $\psi(11\dots1_K)=0101\dots01$ is easily disguised for any deletion pattern $\tau$.

After throwing out the easily disguised candidate codewords, we randomly choose a constant rate code from the remaining candidate codewords.
Careful bookkeeping confirms that with positive probability we obtain a code that decodes $pN$-deletions in the average case.
The bookkeeping is nontrivial, because just as there are a handful of words like $\psi(11\dots 1)$ that are easily disguised as other codewords with deletions, there are also a handful of \emph{easily confused} words like $\psi(KK\dots K)$ that can be confused with many other words when deletions are applied to it.
Furthermore, unlike easily disguised codewords, these easily confused words vary over the different deletion patterns, so we cannot simply throw them out.
However, like for easily disguised codewords, we show the number of easily confused words is an exponentially small fraction of the codebook size in expectation, so such words do not contribute significantly to the number of incorrectly decoded codewords.
Note the subtle difference between easily disguised and easily confused words: a single easily disguised word like $\psi(11\dots1)$ causes \emph{many} candidate codewords to fail to decode under our unique decoding, but any easily confused codeword adds \emph{at most one} failed decoding.

We model managing easily disguised and easily confusable codewords via a directed graph, where, roughly, for each deletion pattern, we consider a graph on $[K]^n$ where $\overrightarrow{YX}$ is an edge if and only if $\tau(\psi(X))\sqsubseteq \psi(Y)$. 
In our proof, we replace the subsequence relation with the matching relation (see \S\ref{sec:5-4}).
In this graph language, the easily disguised codewords correspond to vertices with high outdegree, and the easily confusable codewords correspond to vertices with high indegree.

Our construction illustrates the subtle nature of the oblivious deletion channel and average case errors.
These settings share much of the behavior of the adversarial deletion channel such as the fact that for $p>\half$, every codeword has a deletion pattern destroying all of its information.
Consequently, our approach tackles the oblivious and average case errors using a combinatorial argument just as the best adversarial deletion results do \cite{BukhGH17}.
Yet, the relaxed decoding requirement allows us to exploit it to correct a fraction of deletions approaching 1.


\section{Outline of connection between online and adversarial deletions}
\label{sec:3}
We give a high level outline of the proof of Theorem~\ref{thm:online}. Recall that the goal is to prove that, assuming $p_0\id{adv}<\half$, for some $p < \half$, for every deterministic code of rate bounded away from $0$, there exists a deterministic online deletion channel applying up to $pn$ deletions that prevents successful decoding in the average case. That is, when a uniformly random codeword is transmitted, the probability of incorrect decoding is bounded away from 0. 
Note first that it suffices to find a \emph{randomized} online deletion strategy guaranteeing average-case decoding error because, if for a given code the channel has a random strategy to guarantee average-case decoding error, then sampling a strategy over the randomness of the channel gives a deterministic online strategy that inflicts similar probability of incorrect decoding. 

To do this, we assume that the adversarial zero-rate threshold is $p_0\id{adv}<\half$, and consider $p\sar=1/(3-2p_0\id{adv})$.
Our choice of $p\sar$ satisfies $p_0\id{adv} < p\sar <\half$.
Now we show that for any $p\sar<p<\half$, there exists a randomized strategy that inflicts $pn$ deletions in an online manner to a deterministic code, and guarantees that the probability, over the randomness of a uniformly chosen codeword and the randomness of the adversary, of incorrect decoding is positive (we show it is at least $\frac{1}{10}$).

We adapt the ``wait-push'' strategy of Bassily and Smith \cite{BassilyS14} for an online adversary against erasures.
In the ``wait phase'' of this strategy, the adversary observes a prefix $x\sar$ of $\ell$ bits from the transmitted word $x$ without erasing any bits.
Then it constructs a list $\calL_{x\sar}$ of candidate codewords, among which is the actual codeword chosen by the encoder.
In the ``push phase'', the adversary chooses a codeword $x'$ randomly from $\calL_{x\sar}$, which, with positive probability, corresponds to a different message than the actually transmitted word.
The adversary then, for the last $n-\ell$ bits of the transmission, erases bit $x_i$ of $x$ where $x_i\neq x_i'$ ($x_i'$ the $i$th bit of $x'$).

Our strategy adapts the wait push strategy to deletions.
In our wait phase, we observe a prefix of the string until we know the codeword $x$ \emph{exactly}.
We choose in advance a random bit and delete every instance of that bit that we see during the wait phase.
After the wait phase, suppose we have seen $qn$ bits and deleted $rn$ bits, where $r<q$, and we also know the codeword $x$ exactly.
Using the definition of $p_0\id{adv}$, we can show that, for most choices of $x$, there is another codeword $y$ such that (i) the wait phase lengths $qn$ for $x$ and $y$ are the same, (ii) their majority bits $b$ in the wait phases are the same, (iii) the number of bits $rn$ deleted in the wait phases is the same, and (iv) the length-$(1-q)n$ suffixes $x\sar$ and $y\sar$ of $x$ and $y$, respectively, have a common subsequence $s\sar$ satisfying $|s\sar| > (1-p_0\id{adv})|x\sar|$.
We can form pairs of such codewords $(x,y)$ in advance, so that when we see either $x$ or $y$ we push the suffixes $x\sar$ or $y\sar$ to $s\sar$, so that the received word in both cases is $b^{rn}s\sar$.
For these typical choices of $x$, the total number of required deletions to obtain $b^{rn}s\sar$ is at most $qn/2 + p_0\id{adv}(1-q)n$. 
If $q$ is bounded away from 1 (in our case, less than $(1-p)n$), then we can bound the number of deletions away from $\half$ (in our case, less than $p\sar=1/(3-2p_0\id{adv})$).
On the other hand, if $q$ is large (larger than $(1-p)n$), then we can simply observe the first $(1-p)n$ bits and delete the rest.
Since we can only choose one strategy to run, we choose one of these two strategies at random.
This gives our modified wait push strategy.

\vspace{-1ex}
\section{Preliminaries}
\label{sec:4}
\vspace{-1ex}
  \indent\textbf{General Notation.}
  For a boolean statement $P$, let $\one[P]$ be 1 if $P$ is true and 0 otherwise.

  Throughout the paper, $\log x$ refers to the base-2 logarithm.

  We use interval notation $[a,b] = \{a,a+1,\dots,b\}$ to denote intervals of integers, and we use $[a] = [1,a] = \{1,2,\dots,a\}$.

  For a set $S$ and an integer $a$, let $\binom{S}{a}$ denote the family subsets of $S$ of size $a$.
  Let $U(S)$ denote the uniform distribution on $S$.

  \textbf{Words.}
  A \textit{word} is a sequence of symbols from some \textit{alphabet}.
  We denote string concatenation of two words $w$ and $w'$ with $ww'$.
  We denote $w^k=ww\cdots w$ where there are $k$ concatenated copies of $w$.
  We also denote a concatenation of a sequence of words as $w_1w_2\cdots w_k = \prod_{i=1}^kw_i$.
  We denote words from binary alphabets with lowercase letters $c,s,w$ and words from non-binary alphabets with capital letters $X,Y,Z$.

  A \textit{subsequence} of a word $w$ is a word obtained by removing some (possibly none) of the symbols in $w$.

  A \textit{subword} or \textit{interval} of a word $w$ is a contiguous subsequence of characters from $w$. We identify intervals of words with intervals of integers corresponding to the indices of the subsequence. For example, the interval $\{1,2,\dots,|w|\} = [1,|w|]$ is identified with the entire word $w$.

  Let $w'\sqsubseteq w$ denote ``$w'$ is a subsequence of $w$''.

  Define a \textit{run} of a word $w$ to be a maximal single-symbol subword. That is, a subword $w'$ in $w$ consisting of a single symbol such that any longer subword containing $w'$ has at least two different symbols.
  Note the runs of a word partition the word. For example, $110001$ has 3 runs: one run of 0s and two runs of 1s.

  \textbf{Deletion Patterns.}
  A \textit{deletion pattern} is a function $\tau$ that removes a fixed subset of symbols from words of a fixed length.
  Let $\calD(n,m)$ denote the set of deletion patterns $\tau$ that operate on length $n$ words and apply exactly $m$ deletions.
  For example $\tau:x_1x_2x_3\mapsto x_1x_3$ is a member of $\calD(3,1)$.
  Let $\calD(n) = \cup_{m=0}^n \calD(n,m)$. 

  We identify each deletion pattern $\tau\in\calD(n,m)$ with a size $m$ subset of $[n]$ corresponding to the deleted bits.
  We often use sets to describe deletion patterns when the length of the codeword is understood.
  For example $[n]$ refers to the single element of $\calD(n,n)$.
  Accordingly, let $\subseteq$ be a partial order on deletion patterns corresponding to set inclusion, and let $\abs\tau$ denote the number of bits deleted by $\tau$.
  As such, we have $\tau\subseteq \tau'$ implies $\abs\tau\le\abs{\tau'}$.

  For a word $w$ and $\tau\in \calD(|w|)$, let $\tau(w)$ and $w\setminus\tau$ both denote the result of applying $\tau$ to $w$.
  We use the second notation when we identify sets with deletion patterns, as in the above paragraph where the set elements correspond to the deleted positions.

  In a concatenated code with outer code length $n$ and inner code length $L$, we can identify a deletion pattern on the entire codeword $\tau\in\calD(nL)$ as the ``concatenation'' of $n$ deletion patterns $\tau_1\frown\cdots\frown\tau_n$, one for each inner codeword.
  To be precise, for all $\tau\in\calD(nL)$, there exists $\tau_i\in\calD(L)$ such that $\tau = \cup_{i=1}^n \{j+(i-1)L:j\in\tau_i\}$, and we denote this by $\tau = \tau_1\frown\cdots\frown\tau_n$.
  Using this notation, we refer to $\tau_i$ as an \textit{inner code deletion pattern}.

  \textbf{Graphs.}
  In a (directed or undirected) graph $G$, let $V(G)$ and $E(G)$ denote the vertex set and edge set of $G$ respectively.
  For a subset $W\subseteq V(G)$ of the vertices, let $G\restriction W$ denote the subgraph induced by $W$.
  We use $G\restriction W$ instead of the more standard $G|_{W}$ to avoid too many subscripts.
  For a vertex $v\in V(G)$, let $\deg_G(v)$ denote the degree of $v$ in $G$  when $G$ is undirected, and let $\indeg_G(v),\outdeg_G(v)$ denote the indegree and outdegree of $v$, respectively, in $G$ when $G$ is directed.
  We drop the subscript of $G$ in $\deg,\indeg$ and $\outdeg$ notations when the graph $G$ is understood.

  \textbf{Concentration Bounds.}
  We use the following forms of Chernoff bound. 
  \begin{lemma}[Chernoff]
    \label{lem:chernoff}
    Let $A_1,\dots,A_n$ be independent random variables taking values in $[0,1]$.
    Let $A = \sum_{i=1}^n A_i$ and $\delta\in[0,1]$.
    Then
    \begin{align}
      \Pr[A\le (1-\delta)\E[A]] \ \le \ \exp\pa{-\delta^2\E[A]/2}.
      \label{eq:chernoff}
    \end{align}
    Furthermore, if $A_1,\dots,A_n$ are Bernoulli random variables, then
    \begin{equation}
      \Pr[A\ge (1+\delta)\E[A]] \ \le \ \left( \frac{e^\delta}{(1+\delta)^(1+\delta)} \right)^{\E[A]}
      \label{eq:chernoff-2}
    \end{equation}
  \end{lemma}

  We also use the submartingale form of Azuma's Inequality.
  \begin{lemma}[Azuma]
    \label{lem:azuma}
    Let $c$ be a constant.
    Let $X_1,X_2,\dots$ be a submartingale such that $|X_i - X_{i-1}| \le c$ for all $i$.
    Then for all positive reals $t$, we have
    \begin{equation}
      \Pr[X_k - X_1 \le -t] \le \exp\left( \frac{-t^2}{2c^2(k-1)} \right).
      \label{}
    \end{equation}
  \end{lemma}

\section{Deletion code decoding $p$-fraction of oblivious deletions}
\label{sec:5}


  \subsection{Overview of proof}
  \label{sec:5-1}

  In \S\ref{sec:2} we gave a high level overview.
  We now begin with a brief snapshot of the proof structure and how it is organized.

  We present our general code construction in \S\ref{sec:5-2}.
  The construction uses the ``clean construction'' in \cite{BukhGH17} and an outer code that we choose randomly.
  We analyze properties of the concatenated construction in \S\ref{sec:5-3}.
  We begin by extracting the useful behavior of deletion patterns with respect to our codewords.
  This deletion pattern analysis culminates in Lemma~\ref{lem:del-pattern-3}, allowing us to define the \textit{signature} (Definition~\ref{def:signature}) of a deletion pattern.
  The key result of \S\ref{sec:5-3} is Proposition~\ref{lem:match-3}.
  It states that for any deletion pattern $\tau$, the probability that two random candidate codewords $c,c'$ are confusable (in the sense that $\tau(c)\sqsubseteq c'$) is exponentially small in the code length.
  However, because working with deletion patterns directly is messy, the proposition is written in the language of \textit{matchings}, a technical notion defined in Definition~\ref{def:matching-3}.
  In short, because the inner codewords are nicely behaved, we do not need to know the exact details of the behavior of a given deletion pattern $\tau$, but rather only need certain properties of it, given by its signature.
  We thus define a matching to approximate the subsequence relation using only the signature of $\tau$, so that ``$\sigma(X)$ is matchable in $Y$'' (where $\sigma$ is the outer code deletion pattern given by $\tau$'s signature) holds roughly when ``$\tau(\psi(X))\sqsubseteq \psi(Y)$'' holds.

  Combinatorially, Proposition~\ref{lem:match-3} allows us to finish the proof.
  As stated in \S\ref{sec:2}, for our outer code we consider $[K]^n$ minus a small set of easily disguised candidate codewords. 
  The notion of a easily disguised candidate codeword is well defined by Lemma~\ref{lem:match-4}.
  For a sufficiently small constant $\gamma$, we randomly choose a size $2^{\gamma n}$ outer code over the remaining outer codewords.
  In \S\ref{sec:5-4}, we use the graph language described in \S\ref{sec:2} and prove Lemma~\ref{lem:sparse-directed-whp}, which roughly states that in a sparse directed graph, if we randomly sample a small subset of vertices, the induced subgraph is also sparse (for some appropriate definition of sparse) with high probability.
  Finally, in \S\ref{sec:5-5}, we piece together these results, showing that Lemma~\ref{lem:sparse-directed-whp} guarantees, with positive probability, that our random code combinatorially decodes against $pN$ deletions in the average case.


  \subsection{Construction}
  \label{sec:5-2}


  Let $p\in(0,1)$, and let $\lambda=\lambda(p)$ be the smallest integer such that $(1+p)/2 < 1-2^{-\lambda}$.
  For our argument any $\lambda$ such that $p<1-2^{-\lambda}$ suffices.
  In particular, for $p<\half$, we can choose $\lambda=1$, slightly simplifying the argument as described in \S\ref{sec:2}. 
  However, we choose $\lambda$ to be the smallest $\lambda$ such that $(1+p)/2<1-2^{-\lambda}$ to ensure a clean formula for the rate.

  Let $\delta$ be such that $p = 1-2^{-\lambda} - \delta$.
  Let $n$ be a positive integer.
  With hindsight, choose
  \begin{equation}
    K \ = \ 2^{\ceil{2^{\lambda+5}/\delta}},\qquad
    R \ = \ 4K^4,\qquad
    L \ = \ 2R^{K},\qquad
    N \ = \ nL.
    \label{eq:}
  \end{equation}
  Note that $R$ is even.
  For the remainder of this section, the variables $p,\lambda,\delta,K,R$ and $L$ are fixed.

  In this way we have $1-2^{-\lambda}-\frac{1}{\sqrt{R}}-\frac{\delta}{2} > p$.
  For $i=1,\dots, K$, let $g_i$ be the length $L$ word
  \begin{equation}
    g_i \ = \ \pa{0^{R^{i-1}}1^{R^{i-1}} }^{L/(2R^{i-1})}.
    \label{eq:constr-1}
  \end{equation}
  Consider the encoding $\psi:[K]^*\to\{0,1\}^*$ where $\psi(X_1\cdots X_k)=g_{X_1}g_{X_2}\cdots g_{X_k}$.
  We construct a concatenated code where the outer code is a length $n$ code over $[K]$ and the inner code is $\{g_1,g_2,\dots,g_{K}\}\subseteq\{0,1\}^L$.
  For the outer code, we choose a random code $C_{out}$ where each codeword is chosen uniformly at random from $[K]^n$ minus a small undesirable subset that we specify later.
  We choose our decoding function to be \textit{unique decoding}.
  That is, our decoder iterates over all codewords $c\in C$ and checks if the received word $s$ is a subsequence of $c$.
  If it is a subsequence of exactly one $c$, the decoder returns that $c$, otherwise it fails.
  While this decoder is not optimal in terms of the fraction of correctly decoded codewords (it could try to break ties instead of just giving up), it is enough for this proof.
  Furthermore, since we are showing combinatorial decodability, we do not need the decoder to be efficient.

  If, for some $p'>p$, a code can decode $p'nL$ average case or oblivious deletions, then it can decode $pnL$ deletions.
  Thus, we can decrease $\delta$ until $\delta n$ is an even integer, so may assume without loss of generality that $\delta n$ is an even integer.

  \subsection{Analyzing construction behavior}
  \label{sec:5-3}

  \begin{definition}
    A \textit{$g_i$-segment} is an interval in $[L]$ corresponding to a run of $g_i$.
    Note that the $g_i$-segments partition $[L]$ and are of the form $[1+aR^{i-1},(a+1)R^{i-1}]$ for $a\in\{0,\dots,L/R^{i-1}-1\}$.
  \end{definition}
  Note that $g_i$ has $2R^{K+1-i}$ runs.
  In particular, the number of runs greatly varies between inner codewords.
  This property makes the concatenated construction powerful because it is difficult to find common subsequences of different inner codewords.
  The following definition allows us to reason about the inner codewords in terms of their run counts.
  \begin{definition}
    \label{def:corrupt}
    We say an inner code deletion pattern $\sigma$ \textit{preserves} $g_i$ if $\sigma(g_i)$ has at least $2R^{K+1-i}/\sqrt R = 2R^{K+\half-i}$ runs.
    Otherwise we say $\sigma$ \textit{corrupts} $g_i$.
  \end{definition}

  We start with a basic but useful fact about deletion patterns and runs.
  \begin{lemma}
    \label{lem:runs}
    Suppose $w$ is a word with $r$ runs $I_1,\dots, I_r$ and $\tau$ is a deletion pattern such that $\tau(w)$ has $r'$ runs.
    We can think of these runs $I_k$ as subsets of consecutive indices in $\{1,\dots,|w|\}$.
    Then the number of runs $I_k$ of $w$ completely deleted by $\tau$, i.e. satisfying $I\subseteq \tau$ when $\tau$ is thought of as a subset of $\{1,\dots,|w|\}$, is at least $\frac{r-r'}{2}$.
  \end{lemma}
  \begin{proof}
    Deleting any run reduces the number of runs in a word by at most 2, and $\tau$ reduces the number of runs by $r-r'$, so the claim follows.
  \end{proof}

  The next lemma establishes the usefulness of widely varying runs in our construction.
  It says that even when an inner code deletion pattern has a large number of deletions, most of the inner codewords still look the same in terms of the number of runs.
  The intuition for the lemma is as follows.
  Consider the extreme example of an inner code deletion pattern $\sigma$ that ``completely corrupts'' the inner codewords $g_1,\dots,g_\lambda$.
  That is, $\sigma$ deletes all the zeros of each of $g_1,\dots,g_\lambda$.
  Since $\sigma$ deletes all the zeros of $g_\lambda$, it must delete every other run of $R^{\lambda-1}$ bits, thus deleting $L/2$ bits.
  Applying these deletions alone to $g_{\lambda-1}$ leaves it with half as many runs of length exactly $R^{\lambda-2}$.
  However, since $\sigma$ also deletes all zeros of $g_{\lambda-1}$, it must delete every other run of $R^{\lambda-2}$ bits of the remaining $L/2$ bits, thus deleting $L/4$ more bits.
  Similarly, since $\sigma$ deletes all zeros of $g_{\lambda-2}$, it must delete an additional $L/8$ bits.
  Continuing this logic, we have $\sigma$ must delete $L(1-2^{-\lambda})$ bits total.
  This tells us that if an inner code deletion pattern completely corrupts the inner codewords $g_1,\dots,g_\lambda$, it needs $L(1-2^{-\lambda})$ deletions.
  This logic works even if we chose to corrupt any subset of $\lambda$ inner codewords other than $\{g_1,\dots,g_\lambda\}$.
  One can imagine that corrupting (according to Definition~\ref{def:corrupt}) inner codewords is almost as hard as completely corrupting them, so adding some slack gives the lemma.

  \begin{lemma}
    \label{lem:del-pattern-1}
    If $\sigma$ is an inner code deletion pattern with $\abs{\sigma} \le L\pa{1-\frac{1}{2^\lambda} - \frac{1}{\sqrt{R}}}$, then $\sigma$ preserves all but at most $\lambda-1$ choices of $g_i$. 
  \end{lemma} 
  \begin{proof}
    Suppose $\lambda$ is a positive integer such that $\sigma$ corrupts $g_i$ for $\lambda$ different values of $i$, say $i_1>i_2>\cdots>i_\lambda$.
    We wish to show $\abs\sigma > L\pa{1-\frac{1}{2^\lambda} - \frac{1}{\sqrt{R}}}$.

    Recall that a $g_i$ segment is an interval of the form $[1+aR^{i-1},(a+1)R^{i-1}]$.
    Inductively define the collections of intervals $\calI_1,\dots,\calI_\lambda$ and the sets of indices $I_1,\dots,I_\lambda$ as follows.
    For $1\le a\le \lambda$, set
    \begin{equation}
      \calI_a \ = \ \br{J: \text{$J$ is $g_{i_a}$-segment, $J\subseteq \sigma\setminus\bigcup_{b=1}^{a-1} I_b$}}
      \quad\text{ and }\quad
      I_a \ =\  \bigcup_{J\in \calI_a} J.
    \end{equation}
    Intuitively, $\calI_1$ as the set of runs in $g_{i_1}$ that are entirely deleted by $\sigma$, and $I_1$ is the set of those deleted indices.
    Then, $\calI_2$ is the set of runs of $g_{i_2}$ deleted by $\sigma$ but not already accounted for by $\calI_1$, and $I_2$ is the set of bits in the runs of $\calI_2$.
    We can interpret $\calI_3,\dots, \calI_\lambda$ and $I_3,\dots,I_\lambda$ similarly.
    By construction, $I_1,\dots, I_\lambda$ are disjoint, and their union is a subset of $\sigma$ (thought of as a subset of $[L]$), so $\sum_{b=1}^\lambda|I_b|\le |\sigma|$.
    It thus suffices to prove
    \begin{equation}
      \label{eq:del-pattern-1-goal}
      \sum_{b=1}^\lambda |I_b|\ >\ L\pa{1-\frac{1}{2^\lambda}-\frac{1}{\sqrt{R}}}.
    \end{equation}
    Note that for any $j < j'$, every $g_{j'}$-segment is the disjoint union of $R^{j'-j}$ many $g_{j}$-segments.
    We thus have $[L]$ is the disjoint union of $g_{i_a}$-segments and $I_b$ is also the disjoint union of $g_{i_a}$-segments when $b<a$ (and thus $i_b > i_a$).
    Hence, $[L]\setminus \cup_{b=1}^{a-1}I_{b}$ is the disjoint union of $g_{i_a}$-segments.
    Furthermore, as $R$ is even, each $I_{b}$ covers an even number of $g_{i_a}$-segments, so the segments of $[L]\setminus \cup_{b=1}^{a-1}I_{b}$ alternate between segments corresponding to runs of 0s in $g_{i_a}$ and segments corresponding to runs of 1s in $g_{i_a}$.
    It follows that all runs of $g_{i_a}\setminus \cup_{b=1}^{a-1}I_{b}$ have length exactly $R^{i_a-1}$, so the number of runs in the string $g_{i_a}\setminus \cup_{b=1}^{a-1}I_{b}$ is
    \begin{equation}
      \frac{L}{R^{i_a-1}} - \sum_{b=1}^{a-1} R^{i_{b}-i_a}|\calI_{b}|.
    \end{equation}
    By construction, the only $g_{i_a}$-segments that are deleted by $\sigma$ are the intervals covered by $I_1\cup\cdots\cup I_a$. Since $\sigma$ corrupts $g_{i_a}$, we know $\sigma(g_{i_a})$ has less than $L/(R^{i_a-1}\sqrt{R})$ runs. By Lemma~\ref{lem:runs}, we have
    \begin{equation}
      |\calI_a|\ > \
      \half\pa{\pa{\frac{L}{R^{i_a-1}} - \sum_{b=1}^{a-1} R^{i_{b}-i_a}|\calI_{b}|}-\frac{L}{R^{i_a-1}\sqrt{R}}}.
    \end{equation}
    Simplifying and using $|I_{b}| = R^{i_b-1}|\calI_b|$ for all $b$, we obtain
    \begin{equation}
      \label{eq:del-pattern-1-1}
      |I_a|\ > \ L\pa{\half -\frac{1}{2\sqrt R}} - \half\sum_{b=1}^{a-1}|I_b|.
    \end{equation}
    From here it is easy to verify by induction that, for all $1\le a\le \lambda$, we have
    \begin{equation}
      \label{eq:del-pattern-1-goal-2}
      \sum_{b=1}^a|I_b| \ > \ L\pa{1-\frac{1}{2^a}}\pa{1-\frac{1}{\sqrt{R}}}.
    \end{equation}
    Indeed, \eqref{eq:del-pattern-1-1} for $a=1$ provides the base case, and if we know \eqref{eq:del-pattern-1-goal-2} for some $a-1$, then by \eqref{eq:del-pattern-1-1} we have
    \begin{align}
      \sum_{b=1}^a|I_b|
        \ &> \ L\pa{\half -\frac{1}{2\sqrt R}} + \half\sum_{b=1}^{a-1}|I_b| \nonumber\\
        \ &> \ L\pa{\half -\frac{1}{2\sqrt R}} + \frac{L}{2}\pa{1-\frac{1}{2^{a-1}}}\pa{1-\frac{1}{\sqrt{R}}}\nonumber\\
        \ &= \ L\pa{1-\frac{1}{2^a}}\pa{1-\frac{1}{\sqrt{R}}},
    \end{align}
    completing the induction.
    The induction proves \eqref{eq:del-pattern-1-goal}, from which we have
    \begin{equation}
      \label{eq:del-pattern-1-goal-3}
      |\sigma| \ \ge \ \sum_{b=1}^\lambda |I_b|\ >\ L\pa{1-\frac{1}{2^\lambda}-\frac{1}{\sqrt{R}}},
    \end{equation}
    as desired.
  \end{proof}

  Lemma~\ref{lem:del-pattern-1} motivates the following definition. 
  \begin{definition}
    We say an inner code deletion pattern $|\sigma|$ is \textit{$\ell$-admissible} if $|\sigma|\le L(1-1/2^{\ell+1}-\frac{1}{\sqrt{R}})$.
    \label{def:admissible}
  \end{definition}
  If, for some $\ell$, $\sigma$ is $\ell$-admissible, then Lemma~\ref{lem:del-pattern-1} tells us $\sigma$ corrupts at most $\ell$ different $g_i$.
  However, note that $\ell$-admissibility is stronger that corrupting at most $\ell$ different $g_i$ as $\ell$-admissibility gives a stronger upper bound on the number of deletions in $\sigma$, which is necessary in Lemma~\ref{prop:match-1}.

  \begin{lemma}
    \label{lem:del-pattern-2}
    Let $\delta>0$.
    Let $\tau=\tau_1\frown\cdots\frown\tau_n$ be a deletion pattern with at most $(1-\frac{1}{2^\lambda}-\frac{1}{\sqrt R} - \frac{\delta}{2})N$ deletions.
    There are at least $\delta n$ indices $i$ such that $\tau_i$ is $(\lambda-1)$-admissible.
  \end{lemma}
  \begin{proof}
    By a simple counting argument, we have $|\tau_i| > L(1-\frac{1}{2^{\lambda}}-\frac{1}{\sqrt{R}})$ for at most
    \begin{equation}
      \frac{n\cdot L\pa{1-\frac{1}{2^{\lambda}}-\frac{1}{\sqrt{R}}-\frac{\delta}{2}}}{L\pa{1-\frac{1}{2^{\lambda}}-\frac{1}{\sqrt{R}}}}
        \ \le \ n\pa{1 - \frac{\delta/2}{1-2^{-\lambda}}} 
        \ \le \ n\pa{1-\delta}
    \end{equation}
    values of $i$. For the remaining at least $\delta n$ values of $i$, we have $\tau_i$ is $(\lambda-1)$-admissible.
  \end{proof}
  
  The following corollary allows us to reduce our analysis of a deletion pattern $\tau$ to analyzing positions where $\tau$'s inner code deletion pattern is $(\lambda-1)$-admissible.
  We effectively assume that our deletion pattern completely deletes all inner codewords with a non-admissible index.
  \begin{lemma}
    Let $\tau\in\calD(nL, pnL)$.
    There exists $\tau'\in\calD(\delta nL)$, $\sigma\in\calD(n,(1-\delta)n)$ and sets $S_1,\dots,S_{\delta n}$ such that
    \begin{enumerate}
      \item $|S_i| = \lambda-1$ for all $i$,
      \item for all $X\in[K]^n$, we have $ \tau'(\psi(\sigma(X)))\sqsubseteq \tau(\psi(X))$, and 
      \item when we write $\tau'=\tau'_1\frown\cdots\frown\tau'_{\delta n}$ as the concatenation of $\delta n$ inner code deletion patterns, we have, for all $i$ and all $j\notin S_i$, that $\tau'_i\in\calD(L)$ preserves $g_j$.
    \end{enumerate}
    \label{lem:del-pattern-3}
  \end{lemma}
  \begin{proof}
    Let $\tau=\tau_1\frown\cdots\frown\tau_n$.
    By Lemma~\ref{lem:del-pattern-2}, there exist $\delta n$ indices $\ell_1<\cdots<\ell_{\delta n}$ such that, for $i=1,\dots,\delta n$, $\tau_{\ell_i}$ is $(\lambda-1)$-admissible.
    Choose $\sigma\in\calD(n,(1-\delta)n)$ via $\sigma(X_1\dots X_n) = X_{\ell_1}X_{\ell_2}\dots X_{\ell_{\delta n}}$, and choose $\tau'\in\calD(\delta nL)$ via $\tau' = \tau_{\ell_1}\frown\cdots\frown\tau_{\ell_{\delta n}}$.
    We have $\tau'\circ\psi\circ\sigma(X) \sqsubseteq \tau\circ\psi(X)$ for all $X\in[K]^n$ because $\tau'\circ\psi\circ\sigma(X)$ is simply the result of deleting the remaining bits in inner codewords of non-admissible indices in $\tau\circ\psi(X)$.
    By construction, each $\tau_{\ell_i}\in\calD(L)$ is $(\lambda-1)$-admissible, so we can choose $S_1,\dots,S_{\delta n}$ by setting $S_\ell$ to be the set that $\tau_{\ell_i}$ corrupts.
    Note that some $S_i$ may have size less than $\lambda-1$, but we can arbitrarily add elements of $[K]$ to $S_i$ until it has $\lambda-1$ elements.
    This is okay as item 3 in the corollary statement remains true if we add elements to $S_i$.
  \end{proof}

  In our analysis, for a deletion pattern $\tau = \tau_1\frown\cdots\frown\tau_{n}$, we only care about the behavior of a given inner code deletion pattern $\tau_i$ as far as the set $S_i$ of inner codewords $g_j$ that it corrupts; instead of considering all possible deletion patterns $\tau\in\calD(nL)$, it suffices to only consider all possible $\sigma, S_1,\dots, S_{\delta n}$.
  This motivates the following definition.
  \begin{definition}
    \label{def:signature}
    The \textit{signature} of a deletion pattern $\tau$ is $(\sigma, S_1, S_2,\dots, S_{\delta n})$, where $\sigma, S_1,\dots, S_{\delta n}$ are given by Lemma~\ref{lem:del-pattern-3}.
    If the choice of $\sigma, S_1,\dots, S_{\delta n}$ satisfying the conditions of Lemma~\ref{lem:del-pattern-3} are not unique, then choose one such collection of $\sigma,S_1,\dots,S_{\delta n}$ arbitrarily and assign this collection to be the signature of $\tau$.
  \end{definition}
  Below we define the matchability relation $\prec$.
  Definition~\ref{def:matching-3} allows us to worry only about the signature of a deletion pattern $\tau$ rather than $\tau$ itself.
  Intuitively, we can think of the matchability relation $\prec$ as an approximation of the subsequence relation $\sqsubseteq$.
  Proposition \ref{prop:match-1} establishes this relationship formally.
  Specifically, it states that if $\tau$ is a deletion pattern with signature $(\sigma, S_1,\dots, S_{\delta n})$, then for $X,Y\in[K]^n$, we have $\sigma(X)\sqsubseteq Y$ implies that $\sigma(X)$ has a matching in $Y$ with appropriate parameters.
  This means that if we want to show there are few incorrectly decoded codewords in a given code, it suffices to show that few codewords have an appropriately parameterized matching in some other codeword. 

  We first define type-A and type-B pairs of indices $(i,j)\in\{1,\dots,|X|\}\times\{1,\dots,|Y|\}$.
  Intuitively, pairs $(i,j)$ are type-B only if $\tau_i(\psi(X_i))$ has many more runs than $\psi(Y_j)$, i.e. it is difficult to find (contiguous) subwords of $\tau_i(\psi(X_i))$ as subsequences of $\psi(Y_j)$.
  \begin{definition}
    Let $X,Y\in[K]^*$ be words over the alphabet $[K]$ and let $S_1,\dots,S_{|X|}$ be subsets of $K$.
    Given a pair $(i,j)\in\{1,\dots,|X|\}\times\{1,\dots,|Y|\}$, we say $(i,j)$ is \textit{type-A} with respect to $X,Y,S_1,\dots,S_{|X|}$ (or simply \textit{type-A} if the parameters are understood) if $X_i\in S_i$ or $X_i\ge Y_j$.
    Call a pair $(i,j)$ \textit{type-B} with respect to $X,Y,S_1,\dots,S_{|X|}$ otherwise.
    \label{def:matching}
  \end{definition}

  \begin{definition}
    Let $X,Y\in[K]^*$ be words over the alphabet $[K]$, let $S_1,\dots,S_{|X|}$ be subsets of $K$, and let $s$ and $t$ be positive integers.
    The following algorithm constructs the $(s,t,S_1,\dots,S_{|X|})$ \textit{matching} of $X$ with $Y$.
    Begin with a pair $(a,b)=(1,1)$.
    The first and second coordinates correspond to indices of the strings $X$ and $Y$, respectively.
    Define an \textit{A-move} to be incrementing the first coordinate, $a$, by 1, and a \textit{B-move} to be incrementing of the second coordinate, $b$, by 1.
    \begin{enumerate}
      \item If $a=|X|$ or $b=|Y|$, stop.
      \item Do one of the following
      \begin{enumerate}
        \item If the last $s$ moves were A-moves, make a B-move.
        \item Else if the last $t$ moves were B-moves, make an A-move.
        \item Else if $(a,b)$ is type-A, make an A-move.
        \item Else, $(a,b)$ must be type-B, in which case make a B-move.
      \end{enumerate}
      \item Repeat from step 1.
    \end{enumerate}
    Note that at the end of this algorithm, exactly one of $a=|X|$ and $b=|Y|$ is true.
    We say this matching is a \textit{success} if we ended with $a=|X|$, otherwise it is a \textit{failure}.
    \label{def:matching-2}
 \end{definition}
 \begin{definition}
    Note also that the matching is uniquely determined by $X,Y,S_1,\dots,S_{|X|},s,t$.
    If this matching is a success, we say $X$ is \textit{$(s,t,S_1,\dots,S_{|X|})$-matachable} (or has a \textit{$(s,t,S_1,\dots,S_{|X|})$-matching}) in $Y$, denoted
    \begin{equation}
      X\prec_{(s,t,S_1,\dots,S_{|X|})} Y.
      \label{}
    \end{equation}
    \label{def:matching-3}
  \end{definition}

  \begin{proposition}
    \label{prop:match-1}
    Let $S_1,\dots, S_{\delta n}$ be subsets of $[K]$ of size exactly $\lambda-1$.
    Let $\tau = \tau_1\frown\cdots\frown\tau_{\delta n}\in\calD(nL)$ be a deletion pattern such that for all $i$, $\tau_i\in\calD(L)$ is $(\lambda-1)$-admissible and in particular preserves $g_j$ for all $j\notin S_i$.
    Suppose we have $X\in[K]^{\delta n}$ and $Y\in[K]^n$ such that $\tau(\psi(X))\sqsubseteq \psi(Y)$ (recall $\sqsubseteq$ is the subsequence relation).
    Then $X\prec_{(2^{\lambda},\sqrt{R},S_1,\dots,S_{\delta n})} Y$.
  \end{proposition}
  \begin{proof}
    Let $s=2^\lambda, t = \sqrt{R}$.
    Run the matching algorithm defined above to obtain a matching of $X$ and $Y$.
    Let $\calM$ be the set of all $(a,b)$ reached by some step of the algorithm.
    We wish to show this matching is a success, i.e. that there exists some $b$ such that $(|X|,b)\in\calM$, or, equivalently, there does not exist $a$ such that $(a,|Y|)\in \calM$.

    Since $\tau(\psi(X))\sqsubseteq \psi(Y)$, we can find $\tau(\psi(X))$ as a subsequence of $\psi(Y)$ by greedily matching the bits of $\tau(\psi(X))$ with the bits of $\psi(Y)$.
    Let $\calN$ be the set of $(i,k)$ such that some bit of $\tau_i(\psi(X_i))$ is matched with some bit in $\psi(Y_k)$. 
    We first establish some basic facts about $\calM,\calN$.
    \begin{fact}
      \label{fact:match}
      \begin{enumerate}
        \item $(|X|,|Y|)\notin \calM$.
        \item $\calM,\calN\subseteq \{1,\dots,|X|\}\times \{1,\dots,|Y|\}$.
        \item For all $a\sar\in\{1,\dots,|X|\}, b\sar\in\{1,\dots,|Y|\}$, we have $\{b:(a\sar,b)\in\calM\}$ and $\{a:(a,b\sar)\in\calM\}$ are intervals of consecutive integers of lengths at most $t+1$ and $s+1$, respectively.
        \item For all $i\sar\in\{1,\dots,|X|\}, k\sar\in\{1,\dots,|Y|\}$, we have $\{k:(i\sar,k)\in\calN\}$ and $\{i:(i,k\sar)\in\calN\}$ are intervals of consecutive integers.
        \item Let $\le_\calM$ be a relation on $\calM$ such that $(a,b)\le_\calM(a',b')$ iff $a\le a'$ and $b\le b'$.
          Then $\calM$ is totally ordered under $\le_\calM$.
          As such, we can define $\next_\calM(a,b)$ and $\prev_\calM(a,b)$ to be the next larger and next smaller element after $(a,b)$ under $\le_\calM$, respectively.
          Then $\next_\calM(a,b)\in\{(a+1,b),(a,b+1)\}$ and $\prev_\calM(a,b)\in\{(a-1,b),(a,b-1)\}$.
        \item Let $\le_\calN$ be a relation on $\calN$ such that $(i,k)\le_\calN(i',k')$ iff $i\le i'$ and $k\le k'$.
          Then $\calN$ is totally ordered under $\le_\calN$.
          As such, we can define $\next_\calN(i,k)$ and $\prev_\calN(i,k)$ to be the next larger and next smaller element after $(i,k)$ under $\le_\calN$, respectively.
          Then $\next_\calN(i,k)\in\{(i+1,k),(i,k+1),(i+1,k+1)\}$ and $\prev_\calN(i,k)\in\{(i-1,k),(i,k-1),(i-1,k-1)\}$.
        \item If $(a,b),(a',b')\in\calM$, we never have both $a<a'$ and $b'<b$.
        \item If $(i,k),(i',k')\in\calN$, we never have both $i<i'$ and $k'<k$.
      \qedhere
      \end{enumerate}
    \end{fact}
    For $a\in\{1,\dots,|X|\}$ and $b\in\{1,\dots,|Y|\}$, define
    \begin{align}
      &\alpha_0(b) = \min\{a:(a,b)\in\calM\} &\qquad& \alpha_f(b) = \max\{a:(a,b)\in\calM\} \nonumber\\
      &\beta_0(a) = \min\{b:(a,b)\in\calM\} &\qquad& \beta_f(a) = \max\{b:(a,b)\in\calM\} \nonumber\\
      &\iota_0(b) = \min\{i:(i,b)\in\calN\} &\qquad& \iota_f(b) = \max\{i:(i,b)\in\calN\} \nonumber\\
      &\kappa_0(a) = \min\{k:(a,k)\in\calN\} &\qquad& \kappa_f(a) = \max\{k:(a,k)\in\calN\}.
      \label{}
    \end{align}
    \begin{figure}[t!]
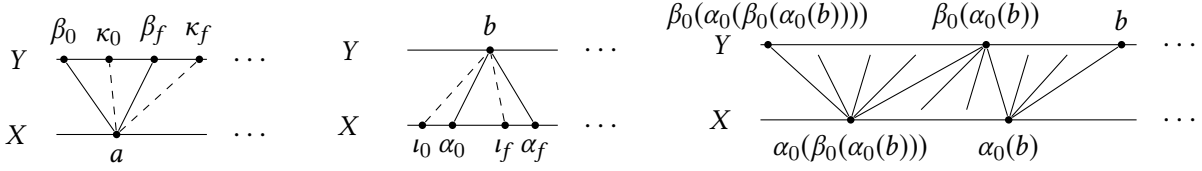

      \begin{subfigure}[b]{0.27\textwidth}
        \tikzcenter{
          \draw (0,1) -- (2,1);
          \draw (0,0) -- (2,0);
          \draw (-0.25,0) node[left] {$X$};
          \draw (-0.25,1) node[left] {$Y$};
          \draw (3,0) node[left] {$\dots$};
          \draw (3,1) node[left] {$\dots$};
          \draw (0.8,0) \ptbelow{a}{$a$};
          \draw (0.1,1) \ptabove{1}{$\beta_0 $};
          \draw (0.7,1) \ptabove{2}{$\kappa_0$};
          \draw (1.3,1) \ptabove{4}{$\beta_f $};
          \draw (1.9,1) \ptabove{3}{$\kappa_f$};
          \draw (a)--(1);
          \draw[dashed] (a)--(2);
          \draw[dashed] (a)--(3);
          \draw (a)--(4);
        }
      \end{subfigure}
      \begin{subfigure}[b]{0.27\textwidth}
          \tikzcenter{
            \draw (0,1) -- (2,1);
            \draw (0,0) -- (2,0);
            \draw (-0.5,0) node[left] {$X$};
            \draw (-0.5,1) node[left] {$Y$};
            \draw (3,0) node[left] {$\dots$};
            \draw (3,1) node[left] {$\dots$};
            \draw (0.2,0) \ptbelow{i0}{$\iota_0$};
            \draw (0.6,0) \ptbelow{a0}{$\alpha_0$};
            \draw (1.7,0) \ptbelow{af}{$\alpha_f$};
            \draw (1.3,0) \ptbelow{if}{$\iota_f$};
            \draw (1.1,1) \ptabove{b}{$b$};
            \draw[dashed] (i0)--(b);
            \draw (a0)--(b);
            \draw (af)--(b);
            \draw[dashed] (if)--(b);
          }
      \end{subfigure}
      \begin{subfigure}[b]{0.45\textwidth}
        \tikzcenter{
          \draw (0,1) -- (5,1);
          \draw (0,0) -- (5,0);
          \draw (-0.25,0) node[left] {$X$};
          \draw (-0.25,1) node[left] {$Y$};
          \draw (6,0) node[left] {$\dots$};
          \draw (6,1) node[left] {$\dots$};
          \draw (1.2,0) \ptbelow{abab}{$\alpha_0(\beta_0(\alpha_0(b)))$};
          \draw (2.0,0) node[draw=none] (b1) {};
          \draw (2.7,0) node[draw=none] (b2) {};
          \draw (3.3,0) \ptbelow{ab}{$\alpha_0(b)$};
          \draw (0.1,1) \ptabove{babab}{$\beta_0(\alpha_0(\beta_0(\alpha_0(b))))$};
          \draw (0.7,1) node[draw=none] (a1) {};
          \draw (1.5,1) node[draw=none] (a2) {};
          \draw (2.2,1) node[draw=none] (a3) {};
          \draw (3.0,1) \ptabove{bab}{$\beta_0(\alpha_0(b))$};
          \draw (3.6,1) node[draw=none] (c1) {};
          \draw (4.2,1) node[draw=none] (c2) {};
          \draw (4.8,1) \ptabove{b}{$b$};
          \draw (babab)--(abab);
          \draw (abab)--(a1);
          \draw (abab)--(a2);
          \draw (abab)--(a3);
          \draw (abab)--(bab);
          \draw (bab)--(b1);
          \draw (bab)--(b2);
          \draw (bab)--(ab);
          \draw (ab)--(c1);
          \draw (ab)--(c2);
          \draw (ab)--(b);
        }
      \end{subfigure}
      \caption{Illustrations of $\alpha_0(b), \alpha_f(b), \beta_0(a),\beta_f(a), \iota_0(b), \iota_f(b),\kappa_0(a),\kappa_f(a)$}
      \label{fig:0}
    \end{figure}
    See Figure~\ref{fig:0} for illustrations of the behavior of these eight functions.
    We first establish a few facts about the notation $\alpha_0,\beta_0,\dots$ that are helpful for developing intuition and are also useful later.
    These proofs are more involved than Fact~\ref{fact:match} and are provided.
    \begin{lemma}
      \label{lem:proper}
      \begin{enumerate}
        \item For all $a\in\{1,\dots,|X|\}$, if $\kappa_f(a)-\kappa_0(a) < \sqrt{R}$, then there exists $b'\in[\kappa_0(a),\kappa_f(a)]$ such that $(a,b')$ is type-A.
        \item For all $a\in\{1,\dots,|X|\}$ we have $\beta_f(a) - \beta_0(a) \le \sqrt{R}$ and for all $\beta_0(a)\le b' < \beta_f(a)$ we have  $(a,b')$ is type-B 
        Furthermore, if $\beta_f(a)-\beta_0(a) < \sqrt{R}$, then $(a,\beta_f(a))$ is type-A.
        \item For all $b\in\{1,\dots,|Y|\}$, we have $\iota_f(b) - \iota_0(b) \le 2^{\lambda}$.
        \item For all $b\in\{1,\dots,|Y|\}$, we have $(i', b)$ is type-A for all $i'\in [\iota_0(b) + 1, \iota_f(b) - 1]$.
        \item For all $b\in\{1,\dots,|Y|\}$, we have $\alpha_f(b)-\alpha_0(b)\le 2^{\lambda}$ and for all $a'\in[\alpha_0(b), \alpha_f(b)-1]$ we have $(a',b)$ is type-A.
          Furthermore, if $\alpha_f(b)-\alpha_0(b) < 2^{\lambda}$, then $(\alpha_f(b),b)$ is type-B.
      \end{enumerate}
    \end{lemma}
    \begin{proof}
      Parts 2 and 5 follow from Definition~\ref{def:matching-2}.

      For part 1, suppose for contradiction that $(a,b')$ is type-B for all $\kappa_0(a)\le b'\le \kappa_f(a)$.
      Then $X_a\notin S_a$ and $X_a < Y_{b'}$ for all such $b'$.
      This means $\tau_a(\psi(X_a))$ has at least $2R^{K+\half-X_a}$ runs while $\psi(Y_{b'})$ has at most $2R^{K+1-(X_a+1)}$ runs for $\kappa_0(a)\le b'\le \kappa_f(a)$.
      On the other hand, we have
      \begin{equation}
        \tau_a(\psi(X_a)) \sqsubseteq \psi\pa{Y_{\kappa_0(a)}\dots Y_{\kappa_f(a)}}.
        \label{eq:proper-1}
      \end{equation}
      As $\kappa_f(a) - \kappa_0(a) < \sqrt{R}$ this means the right side of \eqref{eq:proper-1} has less than $\sqrt{R}\cdot 2R^{K-X_a} = 2R^{K+\half-X_a}$ runs while the left side has at least that many runs, a contradiction.

      For part 3, suppose for contradiction that $\iota_f(b) - \iota_0(b) - 1 \ge 2^{\lambda}$.
      Since $\psi(Y_b)$ contains $\prod_{i'=\iota_0(b)+1}^{\iota_f(b)-1}\tau_{i'}(\psi(X_{i'}))$ as a strict subsequence ($\psi(Y_b)$ additionally contains at least one bit from each of $\tau_{\iota_0}(\psi(X_{\iota_0}))$ and $\tau_{\iota_f}(\psi(X_{\iota_f}))$), we have
      \begin{equation}
        L + 2
        \ \le \ (\iota_f(b)-\iota_0(b)-1)\cdot\frac{L}{2^\lambda} + 2
        \ \le \ \pa{\sum_{i'=\iota_0+1}^{\iota_f-1}\abs{\tau_{i'}(\psi(X_{i'}))} } + 2 
        \ \le \ \abs{\psi(Y_{b})}
        \ = \ L,
        \label{eq:match-2-1}
      \end{equation}
      a contradiction.

      For part 4, suppose for contradiction that $(i',b)$ is type-B for some $\iota_0(b) < i'<\iota_f(b)$.
      Thus $X_{i'} \notin S_{i'}$ and $X_{i'} < Y_{b}$.
      In particular, $\tau_{i'}(\psi(X_{i'}))$ has at least $2R^{K+\half-X_{i'}}$ runs, which is more than the at-most-$2R^{K-X_{i'}}$ runs of $\psi(Y_b)$.
      However, $\iota_0(b) < i'<\iota_f(b)$, so $(i',b)\in\calN$ and in particular $\tau_{i'}(\psi(X_{i'}))\sqsubseteq\psi(Y_b)$, which is a contradiction.
      Note that $i'=\iota_0(b)$ and $i'=\iota_f(b)$ do not guarantee a contradiction because for such $i'$, some bits of $\tau_{i'}(\psi(X_{i'}))$ might be matched with other inner codewords in $\psi(Y)$.
    \end{proof}

    The following definition of proper indices is introduced for convenience.
    Intuitively, indices of $Y$ are $Y$-proper if the bit-matching $\calN$ consumes corresponding indices of $X$ ``slower'' than in the algorithmic matching $\calM$, and indices of $X$ are $X$-proper if the bit-matching $\calN$ consumes corresponding indices of $Y$ ``faster'' than in the algorithmic matching $\calM$.
    \begin{definition}
      We say an index in $a\in\{1,\dots,|X|\}$ is \textit{$X$-proper} if 
      \begin{equation}
        \beta_0(a)\ \le\ \kappa_0(a), 
        \qquad\beta_f(a)\ \le\ \kappa_f(a)
        \label{}
      \end{equation}
      and we say an index $b\in\{1,\dots,|Y|\}$ is \textit{$Y$-proper} if
      \begin{equation}
        \qquad\iota_0(b)\ \le\ \alpha_0(b),  
        \qquad\iota_f(b)\ \le\ \alpha_f(b).
        \label{}
      \end{equation}
    \end{definition}
    \begin{figure}
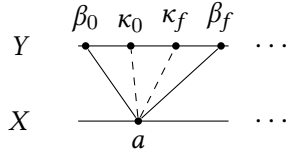

      \begin{center}
        \tikz{
          \draw (0,1) -- (2,1);
          \draw (0,0) -- (2,0);
          \draw (-0.5,0) node[left] {$X$};
          \draw (-0.5,1) node[left] {$Y$};
          \draw (3,0) node[left] {$\dots$};
          \draw (3,1) node[left] {$\dots$};
          \draw (0.8,0) \ptbelow{a}{$a$};
          \draw (0.1,1) \ptabove{1}{$\beta_0$};
          \draw (0.7,1) \ptabove{2}{$\kappa_0$};
          \draw (1.3,1) \ptabove{3}{$\kappa_f$};
          \draw (1.9,1) \ptabove{4}{$\beta_f$};
          \draw (a)--(1);
          \draw[dashed] (a)--(2);
          \draw[dashed] (a)--(3);
          \draw (a)--(4);
        }
      \end{center}
      \caption{Step 2, Assume FSOC $\beta_0(a)\le\kappa_0(a)\le \kappa_f(a) < \beta_f(a)$}
      \label{fig:1}
    \end{figure}
    \begin{claim}
      \label{clm:match-1}
      For all $a\in\{1,\dots,|X|\}$ and all $b\in\{1,\dots,|Y|\}$, $a$ is $X$-proper and $b$ is $Y$-proper.
    \end{claim}
    \begin{remark}
      First we illustrate how the Claim~\ref{clm:match-1} implies Proposition~\ref{prop:match-1}.
      Suppose for contradiction that Proposition~\ref{prop:match-1} is false.
      Then there exists some $a$ such that $(a,|Y|)\in\calM$ and for all $b$ we have $(|X|,b)\notin\calM$.
      In particular, $\alpha_f(b)<|X|$ for all $b\in\{1,\dots,|Y|\}$.
      By the claim, $\iota_f(|Y|)\le \alpha_f(|Y|)<|X|$, implying that no bits from $\tau_{|X|}(\psi(X_{|X|}))$ are matched with bits of $\psi(Y)$, a contradiction of $\tau(\psi(X))\sqsubseteq \psi(Y)$.
      \qedhere
    \end{remark}
    \begin{proof}
      \textbf{Step 1.} First, note that, as $(1,1)\in\calM,\calN$, we have $\alpha_0(1) = \beta_0(1) = \iota_0(1) = \kappa_0(1) = 1$.

      \textbf{Step 2.} 
      Next, we show that for all $a$, if $\beta_0(a)\le\kappa_0(a)$ then $a$ is $X$-proper. 
      That is, we show $\beta_f(a)\le \kappa_f(a)$.
      Suppose for contradiction we have $\kappa_f(a) < \beta_f(a)$ so that $\beta_0(a)\le\kappa_0(a)\le \kappa_f(a) < \beta_f(a)$ (see Figure~\ref{fig:1}).
      By Lemma~\ref{lem:proper} part 2, we have $\beta_f(a)-\beta_0(a)\le\sqrt{R}$ and $(a,b')$ is type-B for all $b'\in[\beta_0(a), \beta_f(a)-1]$.
      In particular, $(a,b')$ is type-B for all $b'\in[\kappa_0(a),\kappa_f(a)]$.
      As $\kappa_f(a) - \kappa_0(a) < \beta_f(a)-\beta_0(a)$, we have $\kappa_f(a) - \kappa_0(a) < \sqrt{R}$, so we can apply Lemma~\ref{lem:proper} part 1 to obtain that $(a,b')$ is type-A for some $b'\in[\kappa_0(a),\kappa_f(a)]$.
      This is a contradiction as all such $b'$ must be type-B.

      \textbf{Step 3.} 
      Next, we show that for all $b$, if $\iota_0(b)\le\alpha_0(b)$ and $\alpha_0(b)$ is $X$-proper, then $b$ is $Y$-proper.
      That is, we show $\iota_f(b)\le\alpha_f(b)$.
      Suppose for contradiction that $\alpha_f(b) < \iota_f(b)$ so that $\iota_0(b)\le\alpha_0(b) \le\alpha_f(b) < \iota_f(b)$ (see Figure~\ref{fig:2}).
      We have $\alpha_f-\alpha_0<\iota_f-\iota_0\le 2^{\lambda}$ by Lemma~\ref{lem:proper} part 3.
      Thus, by Lemma~\ref{lem:proper} part 5, $(\alpha_f(b),b)$ is type-B.
      By Lemma~\ref{lem:proper} part 4, $(i',b)$ is type-A for $i'\in[\iota_0(b)+1,\iota_f(b)-1]$.
      Since $\alpha_f(b)\in[\iota_0(b),\iota_f(b)-1]$ we must have $\alpha_f(b) = \iota_0(b)$, so $\iota_0(b)=\alpha_0(b)=\alpha_f(b)$ (See Figure~\ref{fig:3}).
      By definition of $\alpha_f(b)$, we have $\next_\calM(\alpha_f(b),b) = (\alpha_f(b),b+1)$ so $\beta_f\id{\alpha_f(b)} \ge b+1$.
      However, since we assumed $\alpha_f<\iota_f$, we have $\next_\calN(\alpha_f(b),b) = (\alpha_f(b)+1,b)$, so $\kappa_f(\alpha_f(b)) = b$.
      Thus
      \begin{equation}
        \beta_f(\alpha_0(b)) =  \beta_f(\alpha_f(b)) \ge b+1 > b = \kappa_f(\alpha_f(b)) = \kappa_f(\alpha_0(b)).
        \label{}
      \end{equation}
      On the other hand, $\beta_f\id{\alpha_0(b)}\le\kappa_f\id{\alpha_0(b)}$ by assumption that $\alpha_0(b)$ is $X$-proper, which is a contradiction.
      This covers all possible cases, completing Step 3.

      \begin{figure}
        \begin{subfigure}[b]{0.5\textwidth}
          \tikzcenter{
            \draw (0,1) -- (2,1);
            \draw (0,0) -- (2,0);
            \draw (-0.5,0) node[left] {$X$};
            \draw (-0.5,1) node[left] {$Y$};
            \draw (3,0) node[left] {$\dots$};
            \draw (3,1) node[left] {$\dots$};
            \draw (0.2,0) \ptbelow{i0}{$\iota_0$};
            \draw (0.6,0) \ptbelow{a0}{$\alpha_0$};
            \draw (1.3,0) \ptbelow{af}{$\alpha_f$};
            \draw (1.7,0) \ptbelow{if}{$\iota_f$};
            \draw (1.1,1) \ptabove{b}{$b$};
            \draw[dashed] (i0)--(b);
            \draw (a0)--(b);
            \draw (af)--(b);
            \draw[dashed] (if)--(b);
          }
          \caption{Assume FSOC $\iota_0(b)\le \alpha_0(b)\le \alpha_f(b) < \iota_f(b)$}
          \label{fig:2}
        \end{subfigure}
        \begin{subfigure}[b]{0.5\textwidth}
          \tikzcenter{
            \draw (0,1) -- (5,1);
            \draw (0,0) -- (5,0);
            \draw (-0.5,0) node[left] {$X$};
            \draw (-0.5,1) node[left] {$Y$};
            \draw (6,0) node[left] {$\dots$};
            \draw (6,1) node[left] {$\dots$};
            \draw (0.6,0) \ptbelow{a}{$\iota_0 = \alpha_0 = \alpha_f$};
            \draw (3.3,0) \ptbelow{i}{$\iota_f$};
            \draw (0.1,1) \ptabove{1}{$\beta_0   \id{\alpha_0}$};
            \draw (1.4,1) \ptabove{2}{$\kappa_0  \id{\alpha_0}$};
            \draw (3.1,1) \ptabove{3}{$b=\kappa_f\id{\alpha_0}$};
            \draw (4.8,1) \ptabove{4}{$\beta_f   \id{\alpha_0}$};
            \draw (a)--(1);
            \draw[dashed] (a)--(2);
            \draw[dashed] (a)--(3);
            \draw (a)--(4);
            \draw[dashed] (i)--(3);
          }
          \caption{$\iota_0(b) = \alpha_0(b) = \alpha_f(b) < \iota_f(b)$}
          \label{fig:3}
        \end{subfigure}
        \caption{Step 3}
      \end{figure}

      \begin{figure*}[t!]
        \begin{subfigure}[b]{0.5\textwidth}
          \tikzcenter{
            \draw (0,1) -- (3,1);
            \draw (0,0) -- (3,0);
            \draw (-0.5,0) node[left] {$X$};
            \draw (-0.5,1) node[left] {$Y$};
            \draw (4,0) node[left] {$\dots$};
            \draw (4,1) node[left] {$\dots$};
            \draw (0.3,0) \ptbelow{a}{$\alpha_0(b)$};
            \draw (1.3,0) \ptbelow{i}{$\iota_0(b)$};
            \draw (2.2,0.6) node[draw=none] (a1) {};
            \draw (2.2,0.4) node[draw=none] (i1) {};
            \draw (0.9,1) \ptabove{b}{$b$};
            \draw (2.4,1) \ptabove{k}{$\kappa_f\id{\alpha_0(b)}$};
            \draw (a)--(b);
            \draw[dashed] (a)--(k);
            \draw[dashed] (b)--(i);
            \draw[dashed] (b)--(i1);
            \draw (b)--(a1);
          }
          \caption{$b < \kappa_f(\alpha_0(b))$}
        \end{subfigure}
        \begin{subfigure}[b]{0.5\textwidth}
          \tikzcenter{
            \draw (0,1) -- (3,1);
            \draw (0,0) -- (3,0);
            \draw (-0.5,0) node[left] {$X$};
            \draw (-0.5,1) node[left] {$Y$};
            \draw (4,0) node[left] {$\dots$};
            \draw (4,1) node[left] {$\dots$};
            \draw (0.3,0) \ptbelow{a}{$\alpha_0(b)$};
            \draw (1.3,0) \ptbelow{i}{$\iota_0(b)$};
            \draw (2.2,0.6) node[draw=none] (a1) {};
            \draw (2.2,0.4) node[draw=none] (i1) {};
            \draw (0.9,1) \ptabove{b}{$b=\kappa_f\id{\alpha_0(b)}$};
            \draw[transform canvas={xshift=-1pt}] (a)--(b);
            \draw[dashed,transform canvas={xshift=1pt}] (a)--(b);
            \draw[dashed] (b)--(i);
            \draw[dashed] (b)--(i1);
            \draw (b)--(a1);
          }
          \caption{$b = \kappa_f(\alpha_0(b))$}
        \end{subfigure}
        \caption{Step 4}
        \label{fig:4}
      \end{figure*}

      \textbf{Step 4.}
      We prove that, for all $b\in\{1,\dots,|Y|\}$, if $\alpha_0(b)$ is $X$-proper, then $b$ is $Y$-proper. 
      By Step 3, it suffices to prove $\iota_0(b)\le \alpha_0(b)$.
      Suppose for contradiction that $\alpha_0(b) < \iota_0(b)$.
      Since $(\alpha_0(b),b)\in\calM$, we have $b\le\beta_f(\alpha_0(b))$.
      By assumption, $\alpha_0(b)$ is $X$-proper, so $\beta_f(\alpha_0(b))\le\kappa_f(\alpha_0(b))$, which means $b\le \kappa_f(\alpha_0(b))$.
      If $b< \kappa_f\id{\alpha_0}$, then we have $(\alpha_0(b),\kappa_f(\alpha_0(b)))\in\calN$.
      However, $(\iota_0(b),b)\in\calN$, contradicting Fact~\ref{fact:match} part 8.
      Thus, $\kappa_f(\alpha_0(b)) = b$.
      But then $(\alpha_0(b),b)=(\alpha_0(b),\kappa_f(\alpha_0(b)))\in\calN$ with $\alpha_0(b) < \iota_0(b)$, contradicting the minimality of $\iota_0(b)$.

      \textbf{Step 5.}
      By the same argument as Step 4, we have that, for all $a\in \{1,\dots,|X|\}$, if $\beta_0(a)$ is $Y$-proper, then $a$ is $X$-proper. 

      \textbf{Step 6.}
      We prove by strong induction that for pairs $(a,b)$, ordered by $<_\calM$, we have $a$ is $X$-proper and $b$ is $Y$-proper.
      Combining Steps 1,2, and 3, we have 
      \begin{align}
        \beta_0(1)\ \le\ \kappa_0(1),  
        \qquad\beta_f(1)\ \le\ \kappa_f(1),
        \qquad\iota_0(1)\ \le\ \alpha_0(1),  
        \qquad\iota_f(1)\ \le\ \alpha_f(1),
      \end{align}
      where the first and third inequalities are actually equalities arising from Step 1, the second inequality is established by Step 2, and the fourth is established by Step 3.
      Thus, 1 is both $X$-proper and $Y$-proper.

      Now suppose we have some pair $(a,b)\in\calM$ with $(1,1)<_\calM(a,b)$ and \eqref{eq:proper-1} has been established for all smaller pairs.
      If $\prev_\calM(a,b) = (a-1,b)$, then by the inductive hypothesis, we have $b$ is $Y$-proper.
      However, as $(a-1,b)\in\calM$, we have $(a,b-1)\notin\calM$, so we have $\beta_0(a) = b$.
      Thus $\beta_0(a)$ is $Y$-proper, so $a$ is $X$-proper by Step 5.
      Similarly, if $\prev_\calM(a,b) = (a,b-1)$, then by the inductive hypothesis, we have $a$ is $X$-proper.
      Thus $\alpha_0(b)=a$ is $X$-proper, so $b$ is $Y$-proper by Step 4.
      This completes the proof of Claim~\ref{clm:match-1}, proving Proposition~\ref{prop:match-1}.
    \end{proof}
  \end{proof}

  The following proposition is the key result of this subsection.
  Following our approach, it should be possible to prove this proposition for any choice of $S_1,\dots,S_{\delta n}$, not just $[\lambda-1],\dots,[\lambda-1]$.
  However, because of Lemma~\ref{lem:match-4}, which tells us that $[\lambda-1],\dots,[\lambda-1]$ is the ``worst'' possible choice of $S_1,\dots,S_{\delta n}$, it suffices to prove the proposition as stated.

  \begin{proposition}
    \label{lem:match-3}
    There exists a constant $\beta > 0$ such that for any fixed deletion pattern $\sigma\in\calD(n,(1-\delta)n)$ we have
    \begin{equation}
      \Pr_{X,Y\sim U([K]^n)}[\sigma(X)\prec_{(2^\lambda,\sqrt R,[\lambda-1],\dots,[\lambda-1])} Y] \ < \ 2^{-\beta n}.
      \label{eq:match-3}
    \end{equation}
  \end{proposition}
  \begin{proof}
    Let $s=2^\lambda$.
    With hindsight, let $\beta = \frac{\log K}{16R}$.
    It suffices to prove
    \begin{equation}
      \Pr_{\substack{X\sim U([K]^{\delta n})\\Y\sim U([K]^n)}} \left[X\prec_{(2^\lambda,\sqrt R,[\lambda-1],\dots,[\lambda-1])} Y\right] \ < \ 2^{-\beta n}.
      \label{eq:match-3-1}
    \end{equation}
    This suffices because for any $\sigma$, the distribution of $\sigma(X)$ for $X\sim U([K]^n)$ is the same as $X\sim U([K]^{\delta n})$.

    Let $X_1,\dots,X_{\delta n},Y_1,Y_2,\dots$ be independently chosen from $[K]$.
    Let $X = X_1\dots X_{\delta n}$, $Y = Y_1\dots Y_n$, and $Y_\infty=Y_1Y_2\dots$ so that $X\sim U([K]^{\delta n})$ and $Y\sim U([K]^n)$.
    Construct a $(2^\lambda,\sqrt R,S_1,\dots,S_{\delta n})$-matching of $X$ in $Y_\infty$.
    Note that, as $Y_\infty$ is an infinite random string, $\calM$ succeeds almost surely, i.e. ends in $(|X|,b)$ for some integer $b$.
   
    Let $\calM$ be the set of all reached states $(a,b)$ in the matching, and let $\alpha_f(b) = \max\{a:(a,b)\in\calM\}$ and $\beta_f(a) = \max\{b:(a,b)\in\calM\}$ as in Proposition~\ref{prop:match-1}.
    Let $A_0 = 1, B_0 = 1$.
    For $i\ge 1$, set $A_i = \alpha_f(B_{i-1})$ and $B_i = \beta_f(A_{i-1})$.
    As $A_i - A_{i-1} \le s$ for all $i \ge 1$, we have $A_i$ is well defined for $i\le \delta n/s - 1$. 
    By the definition of matching, $\calM'$ succeeds if and only if $\calM$ succeeds and the final position $(|X|,\beta_f(|X|))$ satisfies $\beta_f(|X|)<|Y|$.
    It thus suffices to prove
    \begin{equation}
      \Pr\left[  B_{\floor{\delta n/s}-1} < |Y| \right] < 2^{-\beta n}.
      \label{}
    \end{equation}

    The key idea of this proof is that the $B_i$'s grow much faster than the $A_i$'s, so that the $B_i$'s ``run out of indices in $Y$'' faster than the $A_i$'s ``run out of indices in $X$'', even though $Y$ is longer than $X$.
    In particular, by definition of a matching, we have $A_{i+1}-A_i\le s=2^\lambda$, but, on the other hand, we show that $B_{i+1}-B_i$ is, in expectation, $\Omega(\log K)$ (see Figure~\ref{fig:10}).

    \begin{figure}
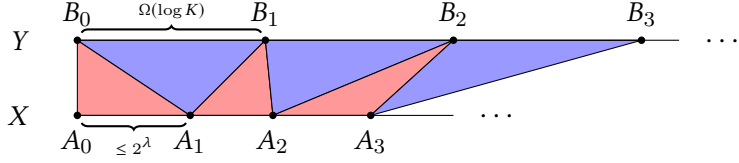

      \tikzcenter{
        \draw (0,1) -- (8,1);
        \draw (0,0) -- (5,0);
        \draw (-0.5,0) node[left] {$X$};
        \draw (-0.5,1) node[left] {$Y$};
        \draw (6,0) node[left] {$\dots$};
        \draw (9,1) node[left] {$\dots$};
        \draw (0,0) \ptbelow{a1}{$A_0$};
        \draw (1.5,0) \ptbelow{a2}{$A_1$};
        \draw (2.6,0) \ptbelow{a3}{$A_2$};
        \draw (3.9,0) \ptbelow{a4}{$A_3$};
        \draw (0,1) \ptabove{b1}{$B_0$};
        \draw (2.5,1) \ptabove{b2}{$B_1$};
        \draw (5,1) \ptabove{b3}{$B_2$};
        \draw (7.5,1) \ptabove{b4}{$B_3$};
        \begin{scope}[on background layer]
          \draw[fill=red!40!white,] (a1.center) -- (b1.center) -- (a2.center) -- cycle;
          \draw[fill=blue!40!white,] (b1.center) -- (a2.center) -- (b2.center) -- cycle;
          \draw[fill=red!40!white,] (a2.center) -- (b2.center) -- (a3.center) -- cycle;
          \draw[fill=blue!40!white,] (b2.center) -- (a3.center) -- (b3.center) -- cycle;
          \draw[fill=red!40!white,] (a3.center) -- (b3.center) -- (a4.center) -- cycle;
          \draw[fill=blue!40!white,] (b3.center) -- (a4.center) -- (b4.center) -- cycle;
        \end{scope}
				\draw [
						thick,
						decoration={
								brace,
								mirror,
								raise=0.1cm
						},
						decorate
				] (a1) -- (a2) 
        node [pos=0.5,anchor=north,yshift=-0.15cm] {\tiny{$\le 2^\lambda$}};
				\draw [
						thick,
						decoration={
								brace,
								raise=0.1cm
						},
						decorate
				] (b1) -- (b2) 
        node [pos=0.5,anchor=south,yshift=0.15cm] {\tiny{$\Omega(\log K)$}};
      }
      \caption{$A_i$'s and $B_i$'s behavior}
      \label{fig:10}
    \end{figure}

    We have two technical lemmas.
    The proofs are straightforward, and we include them in Appendix~\ref{app:3} for completeness.
    \begin{lemma}
      Let $J$ be chosen uniformly from $[K]$.
      Let $D$ be a random variable that is 1 if $J\in[\lambda-1]$ and, conditioned on a fixed $J\ge \lambda$, is distributed as $\min(\Geometric(J/K),\sqrt R)$.
      Then $\E[D]\ge (\log K)/4$.
      \label{lem:geom-1}
    \end{lemma}
    \begin{lemma}
      Let $\lambda'\in[\lambda,K]$ and let $J$ be chosen uniformly from $\{\lambda,\lambda+1,\dots,\lambda'\}.$
      Let $D$ be the random variable that, conditioned on a fixed $J$, is distributed as $\min(\Geometric(J/K),\sqrt R)$.
      Then $\E[D]\ge (\log K)/4$.
      \label{lem:geom-2}
    \end{lemma}
    \begin{claim}
      Let $i\ge 1$.
      For any fixed $A_0,\dots,A_i, B_0,\dots, B_i, X_1,\dots, X_{A_i}, Y_1,\dots, Y_{B_i}$, we have
      \begin{equation}
        \E\left[B_{i+1} - B_i | A_1,\dots, A_i, B_1,\dots, B_i, X_1,\dots, X_{A_i}, Y_1,\dots, Y_{B_i}\right] > \frac{\log K}{4}.
        \label{}
      \end{equation}
    \end{claim}
    \begin{proof}
      It suffices to prove that if we additionally condition on $A_{i+1} - A_i < s$, then the expectation is at least $(\log K)/4$, and that the same is true if we condition on $A_{i+1} - A_i = s$.

      First, note that, for $1\le b<\sqrt{R}$, we have $B_{i+1} = B_i + b$ if and only if $(A_{i+1}, B_i+j)$ is type-B for $j\in\{1,\dots,B_i+b-1\}$ and $(A_{i+1},B_i+b)$ is type-A. 
      If no such $b$ exists, we have $B_{i+1} - B_i = \sqrt{R}$.
      Thus, conditioned on fixed $A_0,\dots,A_{i+1}, B_0,\dots, B_i, X_1,\dots, X_{A_{i+1}}, Y_1,\dots, Y_{B_i}$, we have $B_{i+1}-B_i$ is distributed as $\min(\Geometric(X_{A_{i+1}}/K),\sqrt R)$.

      Suppose we condition on $A_{i+1} - A_i = s$.
      This is equivalent to saying $(A_i+\ell,B_i)$ is type-A for $\ell=1,\dots,s-1$.
      However, this assertion depends only on $X_{A_i}, X_{A_i+1},\dots, X_{A_i+s-1}$, which are independent of $X_{A_{i+1}}$, so we have $X_{A_{i+1}}$ is still distributed uniformly on $[K]$.
      If $X_{A_{i+1}}\in[\lambda-1]$, then $B_{i+1}-B_i=1$, otherwise $B_{i+1}-B_i$ is distributed as $\min(\Geometric(X_{A_{i+1}}/K),\sqrt R)$, so by Lemma~\ref{lem:geom-1} on $D=B_{i+1}-B_i$, we have
      \begin{equation}
        \E[B_{i+1}-B_i|A_{i+1}-A_i=s,A_1,\dots, A_i, B_1,\dots, B_i, X_1,\dots, X_{A_i}, Y_1,\dots, Y_{B_i}] > \frac{\log K}{4}.
        \label{}
      \end{equation}

      Suppose we condition on $A_{i+1}-A_i = s' < s$.
      This is equivalent to saying $(A_i +\ell, B_i)$ is type-A for $\ell=1,\dots,s'-1$ and $(A_i+s', B_i)$ is type-B. 
      This implies $X_{A_{i+1}}\ge \lambda$ (i.e. $X_{A_{i+1}}\notin[\lambda-1]$) and $X_{A_{i+1}}<Y_{B_i}$,
      Since $Y_{B_i}$ is fixed and, without any conditioning, $X_{A_{i+1}}$ is distributed uniformly in $[K]$, we have $X_{A_{i+1}}$ is distributed uniformly on $[\lambda,\dots, Y_{B_i}]$.
      By the previous argument, we have $B_{i+1}-B_i$ is distributed as $\min(\Geometric(X_{A_{i+1}}/K),\sqrt R)$, so by Lemma~\ref{lem:geom-2} on $\lambda'=Y_{B_i}$ and $D=B_{i+1}-B_i$, we have
      \begin{equation}
        \E[B_{i+1}-B_i|A_{i+1}-A_i=s'<s,A_1,\dots, A_i, B_1,\dots, B_i, X_1,\dots, X_{A_i}, Y_1,\dots, Y_{B_i}] > \frac{\log K}{4}.
        \label{}
      \end{equation}

    \end{proof}

    \begin{corollary}
      \label{cl:4}
      We have, for any fixed $B_1,\dots,B_i$,
      \begin{equation}
        \E[B_{i+1}-B_i|B_1,\dots,B_i] > \fourth \log K.
        \label{eq:prop-6}
      \end{equation}
    \end{corollary}

    Continuing the proof of Proposition~\ref{lem:match-3}, let $\alpha = \fourth\log K$ so that
    \begin{equation}
      \alpha (k-1) 
        \ > \ \fourth\log K\cdot\left(\frac{\delta n}{2^\lambda}-1\right)
        \ > \ 2n.
      \label{eq:prop-7}
    \end{equation}
    Define the random variable $B_i' \defeq  B_i - i\cdot \alpha$.
    As $B_i'$ and $B_i$ uniquely determine each other, we have, by Corollary~\ref{cl:4},
    \begin{equation}
      \E[B_{i+1}'-B_i'|B_1',\dots,B_i'] = \E[B_{i+1} - B_i|B_1,\dots,B_i] - \alpha > 0.
      \label{eq:prop-8}
    \end{equation}
    Thus, $B_i'$ form a submartingale with $|B_{i+1}'-B_i'|< \sqrt{R}$, so by Azuma's Inequality (Lemma~\ref{lem:azuma}), we have
    \begin{equation}
      \Pr\left[B_k'-B_1' \le - \frac{\alpha (k-1)}{2}\right] 
      \ \le \ \exp\left( -\frac{(\alpha (k-1)/2)^2}{2(k-1)\cdot (\sqrt{R})^2} \right)
      \ = \ \exp\left( -\frac{\alpha^2(k-1)}{8R} \right)
      \ < \ \exp\left( -\frac{n\log K}{16R} \right).
      \label{eq:prop-9}
    \end{equation}
    Combining with \eqref{eq:prop-7} gives
    \begin{equation}
      \Pr[B_k\le n]
      \ \le \ \Pr[B_k'-B_1' \le -n]
      \ \le \  \Pr\left[B_k'-B_1' \le - \frac{\alpha (k-1)}{2}\right] 
      \ < \ \exp\left( -\frac{n\log K}{16R} \right).
      \label{eq:prop-10}
    \end{equation}
    We conclude
    \begin{align}
      \Pr_{\substack{X\sim U([K]^{\delta n})\\Y\sim U([K]^n)}} [X\prec_{(\sqrt{R},2^\lambda,[\lambda-1],\dots,[\lambda-1])} Y] 
      \ \le \ \Pr[m_{\delta n}\le n] 
      \ \le \ \Pr[B_k\le n]
      \ < \ \exp\left( -\frac{n\log K}{16R} \right).
    \end{align}
    As $\exp(-\frac{n\log K}{16R}) < 2^{-\frac{n\log K}{16R}}$, we have \eqref{eq:match-3-1} is true, completing the proof of Proposition~\ref{lem:match-3}.
  \end{proof}

  To conclude this section, we formalize the intuition that the ``worst possible deletion pattern'', that is, the deletion pattern that makes decoding most difficult in some sense, is the deletion pattern such that each $(\lambda-1)$-admissible inner deletion pattern corrupts $g_1,\dots,g_{\lambda-1}$.
  This fact is intuitive, as it is hard, for example, to find $g_1$ as a subsequence of a random $\psi(Y)$ because $g_1$ has many runs.
  Thus if our deletion pattern $\tau$ corrupts the inner codewords with the most runs, there is a greater probability that over random $X,Y$ we have $\tau(\psi(X))$ is a subsequence of $\psi(Y)$.
  However, note that, to make a clean assertion, we continue to argue using the matching relation rather than subsequence relation.

  \begin{lemma}
    \label{lem:match-4}
    Let $r,s\in\NN$ and $S_1,\dots,S_{\delta n}$ be subsets of $[K]$ of size exactly $\lambda-1$.
    Then for all $Y\in[K]^n$ we have
    \begin{equation}
      \#\left\{X\in[K]^{\delta n}: X\prec_{(2^\lambda,\sqrt R,S_1,\dots,S_{\delta n})} Y\right\} 
      \ \le \ \#\left\{X\in[K]^{\delta n}: X\prec_{(2^\lambda,\sqrt R,[\lambda-1],\dots,[\lambda-1])} Y\right\} 
      \label{eq:match-4}
    \end{equation}
  \end{lemma}
  \begin{proof}
    We start with a claim.
    \begin{claim}
      \label{cl:4-1}
      For every set $A\subseteq[K]$ with size exactly $\lambda-1$, there exists a bijection $h_A:[K]\to[K]$ such that $h_A(x)\in[\lambda-1]$ for $x\in A$ and $h_A(x)\ge x$ for $x\notin A$.
    \end{claim}
    \begin{proof}
      Pair each element $x\in A\setminus[\lambda-1]$ with an element $y\in [\lambda-1]\setminus A$ arbitrarily and for each pair $(x,y)$ set $h_A(x) = y, h_A(y) = x$.
      Note this always gives $x>y$.
      This is possible as $A\setminus[\lambda-1]$ and $[\lambda-1]\setminus A$ have the same size.
      Then set $h_A(x) = x$ for all other $x$.
      It is easy to check this function satisfies $h_A(x)\ge x$ for $x\notin A$.
    \end{proof}

    Fix $Y\in[K]^n$.
    For all $i$, let $h_i:S_i\to [\lambda-1]$ be a bijection such that $h_i(x)\in[\lambda-1]$ for $x\in S_i$ and $f(x)\ge x$ for all other $x$.
    This exists by Claim~\ref{cl:4-1}.
    Let $h:[K]^{\delta n}\to[K]^{\delta n}$ be such that $h(X_1,\dots,X_{\delta n}) = h_1(X_1)h_2(X_2)\cdots h_{\delta n}(X_{\delta n})$.
    Since each of $h_1,\dots,h_{\delta n}$ are bijections, $h$ is a bijection as well.

    Let $X$ be such that $X\prec_{(2^\lambda,\sqrt R,S_1,\dots,S_{\delta n})} Y$.
    We claim $h(X)\prec_{(2^\lambda,\sqrt R,[\lambda-1],\dots,[\lambda-1])} Y$.
    Let $\calM$ be a $(2^\lambda,\sqrt R,S_1,\dots,S_{\delta n})$-matching of $X$ in $Y$ and let $\calM'$ be a $(2^\lambda,\sqrt R,[\lambda-1],\dots,[\lambda-1])$-matching of $h(X)$ in $Y$.
    We know $\calM$ is a successful matching, and we wish to show $\calM'$ is a successful matching.
    
    If $(a,b)$ is type-A with respect to $X,Y,S_1,\dots,S_{|X|}$, then $Y_b \le X_a$ or $X_a\in S_a$, so, by definition of $h$, either $h(X_a)\in[\lambda-1]$ or $Y_b\le X_a\le h(X_a)$, so $(a,b)$ is type-A with respect to $h(X),Y,[\lambda-1],\dots,[\lambda-1]$.
    Let $\ell\ge 0$ be an integer and let $(a,b)\in\calM$ and $(a',b')\in\calM'$ be the state of the matchings after the $\ell$th step of the matchings of $X$ in $Y$ and $h(X)$ in $Y$, respectively. 
    It is easy to see by induction that $a\le a'$ and $b'\le b$; if $a<a'$ then after one move we still have $a\le a'$, and if $a=a'$ and $\calM$s next move is an A-move, then $\calM'$s next move must also be an A-move.
    Since $\calM$ succeeds, we have $(|X|,b)\in\calM$ for some $b<|Y|$, so we conclude that $(|X|,b')\in\calM'$ for some $b'\le b<|Y|$.
    Thus $\calM'$ succeeds as well.
    
    Since $h$ is a bijection such that
    \begin{equation}
      h\pa{\br{X\in[K]^{\delta n}: X\prec_{(2^\lambda,\sqrt{R},S_1,\dots,S_{\delta n})} Y}}
      \ \subseteq \ \br{X\in[K]^{\delta n}: X\prec_{(2^\lambda,\sqrt{R},[\lambda-1],\dots,[\lambda-1])} Y},
      \label{eq:match-4-1}
    \end{equation}
    we have 
    \begin{equation}
      \#\br{X\in[K]^{\delta n}: X\prec_{(2^\lambda,\sqrt{R},S_1,\dots,S_{\delta n})} Y} 
      \ \le \ \#\br{X\in[K]^{\delta n}: X\prec_{(2^\lambda,\sqrt{R},[\lambda-1],\dots,[\lambda-1])} Y} 
      \label{eq:match-4-2}
    \end{equation}
    as desired.
  \end{proof}

  \subsection{Technical combinatorial lemmas}
  \label{sec:5-4}

  \begin{lemma}
    \label{lem:sparse-undirected-whp}
    Let $n$ be a positive integer and suppose we have $0<\beta<1$.
    Suppose that $\gamma=\beta / 3$, $M=2^{\gamma n}$, and $\eps \ge 2^{-(\gamma-o(1)) n/2}$.
    Suppose $G$ is a graph on $N=K^{(1-o(1))n}$ vertices such that each vertex has degree at most $d = Nr$ for some $r = r(n) = 2^{-(\beta-o(1)) n}$.
    Suppose ${C_{out}}$ is chosen as a random subset of $V(G)$ by including each vertex of $V(G)$ in ${C_{out}}$ with probability $M/N$. 
    Then for sufficiently large $n$, we have $\Pr_{C_{out}}[|E(G\restriction {C_{out}})| > \eps {M}]<2^{-\omega(n)}$.
  \end{lemma}
  \begin{remark}
    Essentially this lemma is saying that when the edge density is extremely small, around $2^{-\beta n}$, then for all but an extremely small set of choices for ${C_{out}}$, ${C_{out}}$ is extremely sparse.
  \end{remark}
  \begin{proof}
    Let $E = E(G)$. We know $E$ satisfies $|E| \le dN/2 < N^2r$.

    Enumerate the edges $1,\dots, |E|$ arbitrarily.
    Let $Y_1,\dots, Y_{|E|}$ be Bernoulli random variables denoting whether the $i$th edge is in ${C_{out}}$.
    Let $Z = Y_1+\cdots+Y_{|E|}$. We would like to show
    \begin{equation}
      \Pr_{C_{out}}[Z > \eps {M}] \ < \ 2^{-n}.
    \end{equation}
    We do this by computing a sufficiently large moment of $Z$ and using
    \begin{equation}
      \Pr_{C_{out}}[Z > \eps{M}] \ \le \ \frac{\E[Z^k]}{(\eps {M})^k}.
    \end{equation}
    
    For a tuple of (not necessarily distinct) edges $(e_1,\dots,e_k)$, denote $V(e_1,\dots, e_k)$ to be the set of vertices on at least one of the edges $e_1,\dots,e_k$.
    Alternatively, we say $V(e_1,\dots,e_k)$ is the set of vertices \textit{covered} by edges $e_1,\dots, e_k$.
    Note that for all $e_1,\dots, e_k$, we have $2\le V(e_1,\dots,e_k)\le 2k$.
    Let $P_{k,\ell} = \#\{(e_1,\dots, e_k):|V(e_1,\dots,e_k)|=\ell\}$ denote the number of ordered tuples of edges from $E$ that cover exactly $\ell$ vertices.
    \begin{claim*}
      \begin{equation}
        \label{eq:p_k_ell-fml}
        P_{k,\ell}
          \ < \ N^\ell\cdot r^{\ell/2}\cdot \ell^{2k+\ell-1}\cdot k^k
      \end{equation}
    \end{claim*}
    \begin{proof}
      We bound $P_{k,\ell}$ by first bounding the number of sets (unordered tuples) of edges covering exactly $\ell$ vertices and then multiply by $k!$.
      We first compute the number of ways to construct a forest of trees covering $\ell$ vertices using only edges from $E$.
      We do this by casework on the number of connected components.
      Let $1\le c\le \floor{\ell/2}$ be an integer.
      To construct a forest with $c$ connected components, we first choose $c$ disjoint edges $e_1,\dots,e_c$.
      This can be done in at most $|E|^c$ ways.
      We have $|V(e_1,\dots,e_c)|=2c$.
      For $c+1\le i\le \ell-c$, choose a vertex $v\in V(e_1,\dots, e_{i-1})$ and an edge $vw$ such that $w\notin V(e_1,\dots,e_{i-1})$.
      By construction, for $c+1\le i\le \ell-c$, we have $|V(e_1,\dots,e_i)| = i+c$.
      Thus the $i$\tsup{th} edge can be added in at most $(i+c-1)d$ ways.
      The $i=\ell-c$\tsup{th} edge completes a forest covering $\ell$ vertices.
      Recalling that $|E|<N^2r$ and $d=Nr$, we have that the total number of ways to construct a forest in this fashion is at most
      \begin{equation}
        \label{eq:p_k_ell-forest-bound}
        |E|^c\cdot 2cd\cdot(2c+1)d\cdots(\ell-1) d
        \ \le \ |E|^c\cdot d^{\ell-2c}\cdot \ell^{\ell-2c}
        \ < \ N^\ell\cdot r^{\ell-c}\cdot\ell^{\ell-2c}.
      \end{equation}
      We have used $\ell-c$ edges thus far.
      There are $k-\ell+c$ remaining edges. 
      As $|V(e_1,\dots,e_k)|=\ell$, these remaining edges must connect one of the $\binom\ell 2$ pairs of vertices in $V(e_1,\dots,e_{k-\ell+c})$.
      As $c\le\floor{\ell/2}<\ell$, we have $k-\ell+c < k$.
      The remaining edges can thus be chosen in at most $\binom{\ell}{2}^{k-\ell+c}<\ell^{2k}$ ways.
      Using \eqref{eq:p_k_ell-forest-bound} and multiplying by $k!$ we have
      \begin{equation}
        P_{k,\ell}
          \ < \ k!\cdot \sum_{c=1}^{\ell/2} N^\ell\cdot r^{\ell-c}\cdot \ell^{\ell-2c}\cdot \ell^{2k}
          \ < \ k^k\cdot (\ell/2)\cdot N^{\ell}\cdot r^{\ell/2} \cdot \ell^{\ell-2}\cdot\ell^{2k}
          \ < \ N^\ell\cdot r^{\ell/2}\cdot \ell^{2k+\ell-1}\cdot k^k.
          \qedhere
      \end{equation}
    \end{proof}

    Note that ${M}\sqrt{r}\ll 1$. With this claim, we have
    \begin{align}
      \E[Z^k] 
        \ &= \ \E[(Y_1+\cdots+Y_{|E|})^k] \nonumber\\
        \ &= \ \sum_{(e_1,\dots,e_k)\in\{1,\dots,|E|\}^k} \E[Y_1\cdots Y_k] \nonumber\\
        \ &= \ \sum_{\ell=2}^{2k} \sum_{|V(e_1,\dots,e_k)|=\ell}\E[Y_{e_1}\cdots Y_{e_k}] \nonumber\\
        \ &= \ \sum_{\ell=2}^{2k} \sum_{|V(e_1,\dots,e_k)|=\ell}\pa{\frac{{M}}{N}}^\ell \nonumber\\
        \ &= \ \sum_{\ell=2}^{2k} \pa{\frac{{M}}{N}}^\ell\cdot P_{k,\ell} \nonumber\\
        \ &< \ \sum_{\ell=2}^{2k} \pa{\frac{{M}}{N}}^\ell\cdot N^\ell\cdot r^{\ell/2}\cdot \ell^{2k+\ell-1}\cdot k^k  \nonumber\\
        \ &< \ \sum_{\ell=2}^{2k} \ell^{2k+\ell-1}\cdot k^k  \nonumber\\
        \ &< \ 2k\cdot (2k)^{4k-1}\cdot k^k = 2^{4k}\cdot k^{5k}.
    \end{align}
    Finally, choosing $k=\log n$, we have
    \begin{align}
      \Pr[Z\ge \eps{M}]
        \ \le \ \frac{\E[Z^k]}{(\eps{M})^k}
        \ < \ \frac{2^{4k}k^{5k}}{\eps^k{M}^k}
        \ < \ \pa{\frac{16k^5}{\eps 2^{\gamma n}}}^k 
        \ \le \ \pa{\frac{16k^5}{2^{\gamma n/2}}}^k 
        \ < \ 2^{-\omega(n)}
    \end{align}
    as desired.
  \end{proof}

  \begin{lemma}
    Let $n$ be a positive integer, and suppose $0<\beta < 1$ and $\gamma = \beta / 4$.
    Suppose $M=2^{\gamma n}$ and $\eps \ge 2^{-\gamma n/ 2}$.
    Suppose $G$ is a directed graph on $N = K^{n(1-o(1))}$ vertices such that each vertex has outdegree at most $Nr$ for some $r = 2^{-(\beta-o(1)) n}$.
    Choose a subset $C_{out}\subseteq V(G)$ at random by including each vertex of $V(G)$ in ${C_{out}}$ with probability $M/N$ so that $\E[|{C_{out}}|] = M$.
    Then, for sufficiently large $n$, we have
    \begin{equation}
      \Pr_{{C_{out}}}\left[ \#\left\{ X\in {C_{out}}: \exists Y\in {C_{out}}\suchthat \overrightarrow{YX}\in E(G) \right\} 
      > \eps |{C_{out}}| \right]
      \ < \ 2^{-\omega(n)}.
      \label{eq:sparse-directed-whp}
    \end{equation}
    \label{lem:sparse-directed-whp}
  \end{lemma}
  \begin{proof}
    As, by assumption, $Y$ has outdegree at most $N\cdot r$ for all $Y\in V(G)$, the average indegree of $G$ is at most $N\cdot r$.
    Thus at most $\eps/8$ fraction of all words in $V(G)$ have indegree larger than $\frac8\eps\cdot N\cdot r$.
    Call this set of vertices $W$.
    We have
    \begin{equation}
      \E_{C_{out}}\ba{\#\br{X\in {C_{out}}: \indeg_G(X) > \frac8\eps\cdot N\cdot r}}
        \ = \ \E[|{C_{out}}\cap W|]
        \ \le \ \frac{\eps}{8}M.
    \end{equation}
    Since $|{C_{out}}\cap W|$ is the sum of i.i.d Bernoulli random variables, we have by Lemma~\ref{lem:chernoff}
    \begin{equation}
      \label{eq:graph-1}
      \Pr_{C_{out}}\ba{\#\br{X\in {C_{out}}: \indeg_G(X) > \frac8\eps\cdot N\cdot r } > \frac{\eps}{4}M}
      \ = \ \Pr\ba{|{C_{out}}\cap W|>\frac{\eps}{4} M}
      \ < \ e^{-\half\cdot\frac{\eps}{8}\cdot M} < 2^{-\omega(n)}.
    \end{equation}
    Now consider the undirected graph $H$ on $V(G)$ such that $XY\in E(H)$ if $\overrightarrow{XY}\in E(G)$ and $Y\notin W$ or $\overrightarrow{YX}\in E(G)$ and $X\notin W$.
    For every vertex $v$, the in-edges of $v$ in $G$ correspond to edges in $H$ only if the indegree is at most $\frac{8}{\eps}\cdot N\cdot r$, and the outdegree of $v$ in $G$ is always at most $Nr$.
    Therefore the degree of every vertex in $H$ is at most $r'N$ where $r' = \pa{\frac{8}{\eps}+1}\cdot r < 2^{-(\beta-\gamma-o(1))n}$.

    We are including each vertex of $V(G)$ (and thus each vertex of $V(H)$) in $C_{out}$ independently with probability $M/N$.
    Let $\eps' = \eps/4, \beta' = \frac{3}{4}\beta, \gamma'=\gamma$ so that $\gamma'=\beta'/3, \eps\ge 2^{-(\gamma+o(1)) n/2}, M = 2^{\gamma n},$ and $r' = 2^{-(\beta'-o(1))n}$. 
    By Lemma~\ref{lem:sparse-undirected-whp} for $\beta',\gamma', M,\eps',$ and $r'$, and $H$, we have for sufficiently large $n$ that
    \begin{equation}
      \label{eq:graph-2}
      \Pr_{C_{out}}\pa{E(H\restriction {C_{out}}) \ > \ \frac{\eps}{4} M} < 2^{-\omega(n)}.
    \end{equation}
    Also, by Chernoff bounds, we have
    \begin{equation}
      \label{eq:graph-3}
      \Pr_{C_{out}}\left[|{C_{out}}| \ < \ \frac{3}{4}M\right] < 2^{-\Omega(M)} 
      \implies
      \Pr_{C_{out}}\left[|{C_{out}}| \ < \ \frac{3}{4}M\right] < 2^{-\omega(n)}.
    \end{equation}
    Note that if the number of $X$ such that $\indeg_{G\restriction {C_{out}}}(X) > 0$ at least $\eps |{C_{out}}|$, that is,
    \begin{equation}
      \#\left\{ X\in {C_{out}}: \exists Y\in {C_{out}}\suchthat \overrightarrow{YX}\in E(G) \right\} > \eps |{C_{out}}|,
    \end{equation}
    then one of the following must be false.
    \be
      \ii\label{item:graph-1} 
      $|{C_{out}}| \ge \frac{3}{4}M$
      \ii\label{item:graph-2}
      $\#\{X\in {C_{out}}: \deg_{H\restriction {C_{out}}}(X)>0\} \le \frac{\eps}{2}M$
      \ii\label{item:graph-3}
      $|{C_{out}}\cap W| \le \frac{\eps}{4}M$
    \ee
    Indeed, suppose to the contrary all of these were true.
    The number of vertices in ${C_{out}}\cap W$ with positive indegree in $G\restriction {C_{out}}$ is at most $|{C_{out}}\cap W|\le \frac{\eps}{4}M$.
    If a vertex $X\in {C_{out}}\setminus W$ has positive indegree in $G\restriction {C_{out}}$ then there exists $Y\in {C_{out}}$ such that $\overrightarrow{YX}\in E(G)$, so $XY\in E(H)$ by definition of $H$ and thus $\deg_{H\restriction {C_{out}}}(X)>0$.
    Thus the number of vertices in ${C_{out}}\setminus W$ with positive indegree is at most $\frac{\eps}{2}M$ by property \ref{item:graph-2}.
    Hence the total number of vertices in ${C_{out}}$ with positive indegree in $G\restriction {C_{out}}$ is at most $\frac{3}{4}\eps M \le \eps |{C_{out}}|$.

    Each of items \ref{item:graph-1}, \ref{item:graph-2}, and \ref{item:graph-3} is false with probability $2^{-\omega(n)}$ by \eqref{eq:graph-1}, \eqref{eq:graph-2}, and \eqref{eq:graph-3}, so the probability any of them occur is at most $2^{-\omega(n)}$, as desired.
  \end{proof}

  \subsection{Proof of Construction (Theorem~\ref{thm:average-case})}
  \label{sec:5-5}

  \begin{proof}[Proof of Theorem~\ref{thm:average-case}]
    We would like to show there exists a code $C$ and an $\eps=o_N(1)$ such that, for any deletion pattern $\tau$ deleting $pN$ bits,
    \begin{equation}
      \#\{x\in C: \exists y\in C\,\suchthat\,x\neq y\,\textand\tau(x)\le y\} \ \le \ \eps|C|. 
    \end{equation}
    For $Y\in[K]^n$, define
    \begin{equation}
      \label{eq:fy-def}
      f(Y) \ = \ \Pr_{Z\sim U([K]^{\delta n})}\left[Z\prec_{(2^\lambda,\sqrt{R},[\lambda-1],\dots,[\lambda-1])}Y\right].
    \end{equation}

    Let $\sigma\in\calD(n,(1-\delta)n)$ be a deletion pattern for our outer code $[K]^n$.
    Let $S_1,\dots,S_{\delta n}$ be subsets of $[K]$ of size exactly $\lambda-1$.
    Let $G\sar_{\sigma,S_1,\dots,S_{\delta n}}$ be the graph on the vertex set $[K]^n$ such that $\overrightarrow{YX}$ is an edge if and only if $\sigma(X)\prec_{(2^\lambda,\sqrt{R},S_1,\dots,S_{\delta n})} Y$.
    Note that for all $\sigma\in\calD(n,(1-\delta)n)$ and all $Y\in[K]^n$ we have, by Lemma~\ref{lem:match-4},
    \begin{align}
      \#\{X\in[K]^n:\sigma(X)\prec_{(2^\lambda,\sqrt{R},S_1,\dots,S_{\delta n})}Y\}
      \ &= \ K^{(1-\delta)n} \cdot \#\{Z\in[K]^{\delta n}: Z\prec_{(2^\lambda,\sqrt{R},S_1,\dots,S_{\delta n})} Y\} \nonumber\\
      \ &\le \ K^{(1-\delta)n} \cdot \#\{Z\in[K]^{\delta n}: Z\prec_{(2^\lambda,\sqrt{R},[\lambda-1],\dots,[\lambda-1])} Y\} \nonumber\\
      \ &= \ K^n\cdot f(Y).
    \end{align}
    In the graph language, this means every $Y\in[K]^n$ has outdegree at most $K^n\cdot f(Y)$ in every $G\sar_{\sigma,S_1,\dots,S_{\delta n}}$.

    We remove all $Y$ with large $f(Y)$.
    Let $\beta=\log K/16R$. 
    By Proposition~\ref{lem:match-3}, we know the average value of $f(Y)$ is $2^{-\beta n}$.
    Hence at most $K^n/2$ such $Y$ satisfy $f(Y)\ge 2\cdot 2^{-\beta n} $.
    These are the \textit{easily disguised} candidate codewords mentioned in \S\ref{sec:2}.
    There is thus a set $W\subseteq[K]^n$ such that $|W|=K^n/2$ and each $Y\in W$ satisfies $f(Y) < 2\cdot 2^{-\beta n} = 2^{-(\beta-o(1)) n}$.
    For all $\sigma\in\calD(n,(1-\delta)n), S_1,\dots, S_{\delta n}$, consider the subgraph $G_{\sigma,S_1,\dots,S_{\delta n}} \defeq G\sar_{\sigma,S_1,\dots,S_{\delta n}}\restriction W$.
    By construction, for all $\sigma, S_1,\dots, S_{\delta n}$, the outdegree of every vertex of $G_{\sigma,S_1,\dots,S_{\delta n}}$ is at most $(K^n/2)\cdot r$ for some $r= 2^{-(\beta-o(1)) n}$.

    Let $\gamma = \beta / 4$.
    Let $M=2^{\gamma n}$ and $\eps = 2^{-\gamma n/2}$. 
    Choose a subset $C_{out}\subseteq W$ by including each vertex of $W$ in $C_{out}$ independently with probability $M/|W|$ so that $\E[|C_{out}|]=M$.

    By Lemma~\ref{lem:sparse-directed-whp} for $N=K^n/2, \beta, \gamma, M, \eps, r,$ and $G_{\sigma,S_1,\dots,S_{\delta n}}$, we have that, for all $\sigma$ and sufficiently large $n$,
    \begin{equation}
      \Pr_{C_{out}}\left[ \#\{X\in C_{out}:\exists Y\in C_{out}\suchthat \overrightarrow{YX}\in E(G_{\sigma,S_1,\dots, S_{\delta n}})\} > \eps|C_{out}| \right]
      \ < \ 2^{-\omega(n)}.
    \end{equation}
    This is equivalently
    \begin{equation}
      \Pr_{C_{out}}\left[\#\{X\in C_{out}:\exists Y\in C_{out}\suchthat \sigma(X)\prec_{(2^\lambda,\sqrt{R},S_1,\dots,S_{\delta n})} Y\} > \eps|C_{out}|\right]
      \ < \ 2^{-\omega(n)}.
    \end{equation}
    Union bounding over the at-most-$2^{n}$ possible values of $\sigma$ and $\binom{K}{\lambda-1}^{\delta n}$ values of $S_1,\dots,S_{\delta n}$ gives
    \begin{align}
      \Pr_{C_{out}}\Big[\exists \sigma,S_1,\dots,S_{\delta n} &\suchthat \#\big\{X\in C_{out}:\exists Y\in C_{out}\suchthat \sigma(X)\prec_{(2^\lambda,\sqrt{R},S_1,\dots,S_{\sigma n})} Y\big\} > \eps|C_{out}|\Big] \nonumber\\
      \ &< \ 2^{n}\cdot\binom{K}{\lambda-1}^{\delta n}\cdot 2^{-\omega(n)}\nonumber\\
      \ &< \ 2^{n}\cdot2^{n\delta(\lambda-1)\log K}\cdot 2^{-\omega(n)} = 2^{-\omega(n)}.
    \end{align}
    Additionally, with probability approaching 1, $|C_{out}|\ge \frac34 M$.
    Thus, there exists a $C_{out}$ with $|C_{out}|\ge \frac{3}{4}M$ such that the following holds for all $\sigma\in\calD(n,(1-\delta)n)$ and all $S_1,\dots,S_{\delta n}\in\binom{[K]}{\lambda-1}$: at most $\eps|C_{out}|$ of the codewords $X$ are $(2^\lambda,\sqrt{R},S_1,\dots,S_{\delta n})$ matchable in some other codeword $Y\in C_{out}$.  

    By Lemma~\ref{lem:del-pattern-3}, there exists $\sigma\in D(n, (1-\delta)n)$ and $\tau'\in\calD(\delta nL)$ such that the following holds: If we write $\tau'=\tau'_1\frown\cdots\frown\tau'_{\delta n}$ then each $\tau'_i$ is $(\lambda-1)$-admissible and preserves all $g_j$ for all $j$ not in some size-$\lambda-1$ set $S_i$, and furthermore we have $\tau'(\psi(\sigma(X)))\sqsubseteq\tau(\psi(X))$ for all $X\in[K]^n$.
    By choice of $C_{out}$, for at least $(1-\eps)|C_{out}|$ choices of $X\in C_{out}$, $\sigma(X)$ is not $(2^\lambda,\sqrt{R},S_1,\dots,S_{\delta n})$ matchable in all $Y\in C_{out}$.
    Thus, for these $X$, we have $\tau'(\psi(\sigma(X)))\not\sqsubseteq \psi(Y)$ for all $Y\in C_{out}$ by Lemma~\ref{prop:match-1}.
    Since $\tau'(\psi(\sigma(X)))\sqsubseteq\tau(\psi(X))$ for all $X\in[K]^n$, we have for these $(1-\eps)|C_{out}|$ choices of $X$ that $\tau(\psi(X))\not\sqsubseteq\psi(Y)$ for all $Y\in C_{out}$.
    Thus, choosing $C=\psi(C_{out})$ gives our desired average deletion code.

    Recall $\lambda$ is the smallest integer such that $(1+p)/2<1-2^{-\lambda}$ and $\delta = 1-2^{-\lambda}-p$ so that $2^\lambda = \Theta(1/(1-p))$ and $\delta = \Theta(1-p)$.
    Recall $K = 2^{\ceil{2^{\lambda+5}/\delta}}$ and $R=4K^4$ so that $K=2^{\Theta(1/(1-p)^2)}$.
    The rate of the outer code is $\log 2^{\gamma n}/n\log K$ and the rate of the inner code is $\log K/L$.
    The total rate is thus
    \begin{equation}
      \calR \ = \ \frac{\log K}{48R\cdot 2R^K} \ = \ 2^{-2^{\Theta(1/(1-p)^2)}}.  
      \label{}
    \end{equation}
    This completes the proof of Theorem~\ref{thm:average-case}.
  \end{proof}


  \section{Relating zero-rate threshold of online and adversarial deletions}
  \label{sec:6}

  \subsection{Proof that $p_0\id{adv}=\half$ if and only if $p_0\id{on,d}=\half$}
  \label{sec:6-1}


Recall that the \textit{zero-rate threshold for adversarial deletions}, $p_0\id{adv}$, is the supremum of all $p$ such that there exists a family of code $C\subseteq\{0,1\}^n$ with $|C|=2^{\Omega(n)}$ satisfying $\LCS(C) > (1-p)n$.
Additionally, recall that the \textit{deterministic zero-rate threshold for online deletions}, $p_0\id{on,d}$, is the supremum of $p$ such that there exist families of deterministic codes with rate bounded away from 0 that can correct against $pn$ online deletions in the \emph{average case} (when a uniformly random codeword is transmitted). 
  That is, when a uniformly random codeword is transmitted, the probability of a decoding error, over the choice of the codeword and possibly the randomness of the online channel strategy, is $o(1)$ in the block length of the code.
  Since an online strategy can guess the minority bit in the codeword and delete all its occurrences, adversary can always delete the majority bit, we trivially have
 \begin{equation}
    p_0\id{adv} \le p_0\id{on,d} \le \half.
    \label{fact:on-1}
  \end{equation}
  Recall that the currently best known lower bound on $p_0\id{adv}$ is based on the code construction in \cite{BukhGH17} and implies $p_0\id{adv} \ge \sqrt{2}-1$.  An outstanding question in this area is whether $p_0\id{adv} = 1/2$. Our main result in this section ties this question to the corresponding question for online deletions. Namely, we have
\begin{theorem}
  If $p_0\id{adv}=\half$ if and only if $p_0\id{on,d}=\half$.
  \label{thm:adv-on-0}
\end{theorem}
By virtue of \eqref{fact:on-1},  Theorem~\ref{thm:adv-on-0} follows if we show that $p_0\id{adv}<\half$ implies $p_0\id{on,d}<\half$. 
Our result is a negative result against deterministic codes, and our argument does not immediately extend to stochastic codes. While we expect such a result to exist, we elaborate \S\ref{sec:6-2} on the current limitations of our approach to proving a negative result against stochastic codes.

The idea of the proof of Theorem~\ref{thm:adv-on-0} is based on a suitable adaptation of the ``wait-push'' idea of \cite{BassilyS14}. 
We wait $qn$ bits until we know the codeword, then push so that it is confused with another codeword. In the first wait phase, we delete the minority bit for a budget of $q n/2$ deletions, and while we push we need at most $(1-q)np_0\id{adv}$ deletions. 
If we assume that $p_0\id{adv} <\half$, we end up with a total of $\alpha n$ deletions for some $\alpha<\half$. The exact details are more subtle, but this is the rough idea.

From the definition of $p_0\id{adv}$, we have the following fact.
\begin{fact}
  For any $\calR>0$ and $p > p_0\id{adv}$, there exists $n_0$ such that for all $n\ge n_0$, for any $C\subseteq\{0,1\}^n$ with $|C|\ge 2^{\calR n}$, there exist two strings $x,y\in C$ such that $\LCS(x,y)>(1-p)n$.
  \label{lem:adv-1}
\end{fact}

\begin{corollary}
  For any $\calR>0$ and $p>p_0\id{adv}$, there exists $n_0$ such that for all $n\ge n_0$, for any $C\subseteq\{0,1\}^n$ with $|C|\ge 2^{\calR n}$, there exists $C'\subseteq C$ such that $|C'|\ge |C|-2^{\calR n/2}$, and all elements of $C'$ are confusable with another element of $C$.
  That is, for all $x\in C'$ there exists $y\in C$ such that $x\neq y$ and $\LCS(x,y)>(1-p)n$.
  \label{lem:adv-2}
\end{corollary} 
\begin{proof}
  We can apply Fact~\ref{lem:adv-1} for $\calR'=\calR/2$.
  Start with $C'=\emptyset$.
  While $|C\setminus C'| > 2^{\calR n/2}$, find an $x$ that is confusable with some $y$ in $C$ and add it to $C'$.
  This is possible by Fact~\ref{lem:adv-1}.
\end{proof}

The following is the precise statement of our result relating the zero-rate thresholds for online and adversarial deletions. 
\begin{theorem}
  We have
  \begin{equation}
    p\id{on,d}_0 \le \frac{1}{3-2p\id{adv}_0}.
    \label{eq:adv-on-1}
  \end{equation}
  In particular, if $p\id{adv}_0<\half$, then $p\id{on,d}_0<\half$.
  \label{thm:adv-on-1}
\end{theorem}

\begin{remark}
  If it is the case that the bound in \cite{BukhGH17} is tight and $p_0\id{adv} = \sqrt{2}-1 \approx 0.414$, then Theorem~\ref{thm:adv-on-1} gives $p_0\id{on,d} \le 0.4605$.
\end{remark}
The rest of this section is devoted to proving Theorem~\ref{thm:adv-on-1}. We start with the  following helpful definition.
\begin{definition}
  \label{def:wait}
  Given a code $C$ and a codeword $x\in C$ passing through an online deletion channel, define the \textit{wait phase} and the \textit{push phase} for a codeword as follows.
  Let $q_xn$ be the smallest index such that $x_1\dots x_{q_xn}$ uniquely determines the codeword.
  The wait phase refers to the time until Adv receives bit $q_xn$, but not acted on it (chosen whether to transmit or delete it).
  Thus in the wait phase Adv receives bits $1,\dots,q_xn$ and has acted on bits $1,\dots,q_xn-1$.
  The push phase is the time after the end of the wait phase.
  We say the \textit{wait length} of $x$ is $q_xn$ and the \textit{push length} is $(1-q_x)n$.

  For a codeword $x$, let $q_x$ denote the \textit{relative wait length} (wait length divided by $n$), and let $r_{x,b}$ denote number $b$-bits that appear in $x_1x_2\dots x_{q_xn}$, so that $r_{x,0} + r_{x,1} = q_x$.
  Let $b_x\in\{0,1\}$ denote the bit that appears more frequently in $x_1\dots x_{q_xn}$ (break ties arbitrarily).
\end{definition}
\begin{proof}[Proof of Theorem~\ref{thm:adv-on-1}]
  Since $p_0\id{on,d}\le\half$, \eqref{eq:adv-on-1} holds if $p_0\id{adv}=\half$, so assume $p_0\id{adv}<\half$.

  Let $p$ be such that
  \begin{equation}
    \half > p > \frac{1}{3-2p\id{adv}_0}.
    \label{eq:adv-on-2}
  \end{equation}
  We show that, for large enough $n$, when Adv gets $pn$ online deletions, he can, with constant probability independent of $n$, guarantee that, if the encoder sends $x\in C$, the decoder receives a string $s$ that is a substring of $y\in C$ where $x\neq y$.

  Call a pair of codewords $x,y$ \textit{online-confusable} if $q_x = q_{y}\le 1-p
  $, $r_{x,0} = r_{y,0}$, $b_x=b_{y}$, and 
  \begin{equation}
    \LCS(x[q_xn:],y[q_xn:]) > (1-q_x)(1-p_0\id{adv})n.
    \label{eq:stoc-2-4}
  \end{equation}
  For an online-confusable pair $x,y$, let $s\sar(x,y)$ denote a common subsequence of $x[q_xn:]$ and $y[q_xn:]$ with length at least $(1-q_x)(1-p_0\id{adv})n$, and let $s(x,y)$ denote the string $(b_x)^{r_{x,b}n}s\sar(x,y)$.
  Note that $s(x,y)$ is a common substring of $x$ and $y$ and, by choice of $p$ in \eqref{eq:adv-on-2} has length at least
  \begin{align}
    |s(x,y)| 
    &= r_{x,b}n + (1-q_x)(1-p_0\id{adv})n \nonumber\\
    &\ge \frac{q_xn}{2} + (1-q_x)(1-p_0\id{adv})n  \nonumber\\
    &= \pa{1-p\id{adv}_0 - \pa{\half - p\id{adv}_0}q_x} n \nonumber\\
    &\ge \pa{1-p\id{adv}_0 - \pa{\half - p\id{adv}_0}(1-p)} n \nonumber\\
    &> (1-p)n \ ,
    \label{eq:stoc-2-5}
  \end{align}
  where the last step follows from the bound on $p$ in \eqref{eq:adv-on-2}.
  
\noindent   We now claim we can find many disjoint online-confusable pairs.
  For fixed $q\sar, r\sar, b\sar$, let 
  \begin{equation}
    C(q\sar, r\sar,b\sar) \ \defeq \ \{x\in C: q_x = q\sar, r_{x,b\sar} = r\sar, b_x = b\sar\}.
    \label{eq:stoc-2-6}
  \end{equation}
  By Lemma~\ref{lem:adv-1} on the set $\{x[q\sar n:]: x\in C(q\sar,r\sar,b\sar)\}$, among any $2^{\calR n/2}$ (in fact, any $2^{\calR(1-q\sar)n/2}$) codewords in $C(q\sar,r\sar,b\sar)$, there exists some two codewords $x$ and $y$ in $C(q\sar,r\sar,b\sar)$ such that $\LCS(x[q\sar n:],y[q\sar n:]) > (1-q\sar)(1-p_0\id{adv})n$.
  Thus, among any $2^{\calR n/2}$ codewords in $C(q\sar,r\sar,b\sar)$, there are some two that are online-confusable.
  Thus for every $q\sar,r\sar,b\sar$, we can construct disjoint pairs of codewords $(x,y)$ that are online-confusable, so that all but $2^{\calR n/2}$ of the elements of $C(q\sar,r\sar,b\sar)$ are paired.
  For $(x,y)$ in a pair, call $x$ and $y$ \emph{partners}.
  Thus, all but $2n^2\cdot2^{\calR n/2}$ of the elements of $C$ have a partner.
  For sufficiently large $n$, this means that at least $0.99|C|$ elements have a partner.

  The online channel choose to run one of the following strategies, each with probability $\half$.
  \begin{enumerate}
    \item Strategy 1
      \begin{enumerate}
        \item Choose a bit $b$ uniformly from $\{0,1\}$.
        \item During the wait phase, delete every $x_i=1-b$.
          After the wait phase, we know the codeword $x$ and have transmitted $r_{x,b}n$ copies of the bit $b$.
        \item If $x\notin C(q_x,r_{x,b},b)$ (i.e. $x\in C(q_x,r_{x,1-b},1-b)$) or $x$ has no partner codeword, transmit the remaining bits (i.e. give up). 
        \item If $x\in C(q_x,r_{x,b},b)$ and $x$ has a partner $y$, then let $s_{xy}\sar$ by the canonical common subsequence of $x\sar,y\sar$ with $|s_{xy}\sar| \ge (1-q_x)(1-p_0\id{adv})n$.
          Since we know that the remaining bits sent by the encoder is $x\sar$, we can delete bits so that the decoder receives $s\sar$.
        \item If the online channel reaches $pn$ deletions at any point, transmit the remaining bits.
      \end{enumerate}
    \item Strategy 2
      \begin{enumerate}
        \item Transmit the first $(1-p)n$ bits.
        \item Delete the last $pn$ bits.
      \end{enumerate}
  \end{enumerate}

  The idea behind the strategy is that either the wait phase is short, in which case strategy 1 requires fewer than $n/2$ deletions, or the wait phase is long in which case strategy 2 gives less than $n/2$ deletions.
  In particular with probability $\half$, we employ a strategy (either strategy 1 or strategy 2) that, with constant probability, needs at most $\min(q_xn, q_xn/2+(1-q_x)p_0\id{adv}n)$ deletions, and we check this is at most $pn$.

  Suppose for codeword $x\in C$, the wait length is $q_xn$.
  Condition on the fact that $x$ has a partner $y$ with $q_x=q_y$ and canonical common subsequence $s_{xy}\sar$ of $x\sar$ and $y\sar$ with $|s\sar_{xy}|\ge(1-q_x)(1-p\id{adv}_0)n$.
  By the above argument, this assumption holds with probability more than 0.99.

  If $q_x>1-p$, then knowing $x_1\dots x_{(1-p)n}$ does not uniquely determine the codeword $x$, so there exists another codeword $y\in C$ with the same prefix $y_1\dots y_{(1-p)n}$.
  Then when the encoder sends each of $x$ and $y$, the online channel applying strategy 2 transmits the same string $x_1\dots x_{(1-p)n} = y_1\dots y_{(1-p)n}$.
  Since we assume the choice of codeword in $C$ is uniformly at random, $x$ and $y$ are equally likely to be transmitted, so the decoder fails in this case with probability at least $\half$.
  
  If $q_x\le 1-p$, then with probability $\half$ over the choice of $b$, we have $x\in C(q_x,r_{x,b},b)$, in which case strategy 1 reaches step (d). 
  Since $x\in C(q_x,r_{x,b},b)$, we have $x_1\dots x_{q_xn}$ has majority bit $b$, so step (b) deletes at most $q_xn/2$ bits. 
  Since $|s\sar_{xy}|\ge (1-q_x)(1-p\id{adv}_0)n$, step (d) deletes $|x\sar| - |s_{xy}\sar| \le (1-q_x)p\id{adv}_0n$ bits.
  The total number of bits required by strategy 1 is thus
  \begin{equation}
    \frac{q_xn}{2} + (1-q_x)p\id{adv}_0n 
    = \pa{p\id{adv}_0 +\pa{\half - p\id{adv}_0}q_x} n
    < \pa{p\id{adv}_0 +\pa{\half - p\id{adv}_0}(1-p)}n 
    < pn
    \label{}
  \end{equation}
  by our choice of $p$.
  Because $x$ and $y$ are partners, the encoder receives $b^{q_xn}s\sar_{xy}$ for both $x$ and $y$ being sent through the online channel.
  Again, as each of $x$ and $y$ is equally likely, the decoder fails with probability at least $\half$.

  We have shown that with probability at least 0.99 over the choice of codeword, either strategy 1 chooses the correct value of $b$ with probability $\half$, and thus causes decoding error with probability at least $\fourth$, or strategy 2 causes decoding error with probability at least $\half$.
  Adv chooses the ``correct'' strategy with probability $\half$, so with total probably at least $0.99\cdot\half\cdot\min(\fourth,\half) > \frac{1}{10}$, there is a decoding error, as desired.
\end{proof}
 
\subsection{Extending to stochastic codes}
\label{sec:6-2}

It is natural to wonder whether the above conditional negative result on deterministic coding against online deletions with deletion fraction approaching $1/2$ can be extended to stochastic codes.
Indeed, \cite{BassilyS14} and \cite{DeyJLS13}, provide tight upper bounds for the capacities of the online erasure and substitution channels, respectively.
Both works provide online channels that cause constant probability decoding errors for both deterministic and stochastic codes.
Both use a two-stage strategy.
In \cite{BassilyS14}, the authors use a ``wait-push'' strategy to construct an online erasure adversary that, for some transmitted word $x$, erases none of the first $qn$ bits, for suitable $0<q<1$, then suitably chooses a random codeword $x'$, and erases all bits $x_i$ for which $i>qn$ and $x_i\neq x_i'$. 
In \cite{DeyJLS13}, the authors use a ``babble-push'' strategy to construct an online substitution adversary that, for some transmitted word $x$, randomly flips $q'n$ of the first $qn$ bits for some $0<q'<q<1$, then suitably chooses a random codeword $x'$, and with probability $1/2$ flips each bit $x_i$ for which $i>qn$ and $x_i\neq x_i'$.

A key difference between our wait-push strategy and the strategies of \cite{BassilyS14} and \cite{DeyJLS13} is that they do not require the online adversary to know the transmitted word \emph{exactly} - they simply need a guarantee that, with constant probability, the codeword $x'$ that they are ``pushing to'' (i) has low Hamming distance from the transmitted word and (ii) is encoded by a different message.
However, our wait-push strategy requires us to wait until we know the transmitted codeword exactly, so that we can identify the LCS between the transmitted word $x$ and the word $x'$ we are pushing to.
In the case that the length of the wait phase, $qn$, is large (say, close to $n$), our online channel against deterministic codes can still guarantee constant probability of error by, with probability 1/2, choosing a different strategy that simply deletes the last $pn$ bits.
However, this patch of choosing a different strategy doesn't work in the stochastic case, since a stochastic code can have a message be identified using a small fraction of the bits while individual codewords need many bits to be identified.
In such a code, the wait-push strategy must wait a long time until it knows the codeword exactly and is thus expensive, while the ``delete the last $pn$ bits'' strategy obscures only the codeword but not the message.
Stochastic codes with this description exist in the literature.
One example is the mega-subcode construction of \cite{ChenJL15}, which encodes a message independently into $t$ different stochastic codes and concatenates the results: in this construction, the message can be identified in the first $n/t$ bits, but an online channel must see at least $n(t-1)/t$ bits to identify the codeword.

In spite of these obstacles to obtaining a negative result for stochastic codes, our result for the deterministic case is still strong evidence that the zero rate thresholds for online and adversarial deletions are either both or neither 1/2. 
In the oblivious error model, we simultaneously constructed deterministic codes under an average error criterion and stochastic codes under a worst case error criterion for $p$ approaching 1, suggesting that for the online setting the zero-rate threshold and deterministic zero-rate thresholds are the same.
As a caution, we note that this is not always the case \cite{DeyJLS16}: for online erasure channels with one bit of delay, the capacity of the channel for deterministic codes under average error criterion is $1-2p$, while for stochastic codes it is $1-p$.


\section{Conclusion}

  \subsection{Decoding deletions vs. insertions and deletions}

  A lemma due to Levenshtein \cite{Levenshtein66} states that a code $C$ can decode against $pn$ adversarial deletions if and only if it can decode against $pn$ adversarial insertions and deletions. 

  In the oblivious error model, decoding insertions vs insertions and deletions are not the same.
  In decoding against oblivious deletions, we not only need to worry about \textit{whether} two codewords are confusable by insertions and deletions, but also \textit{the number of ways} in which they are confusable.
  In particular, when the fraction of oblivious deletions is greater than $\half$, it is inevitable that there are codewords $c$ and $c'$ that are confusable, but we need the number of deletion patterns $\tau$ that confuse them, i.e. with $\tau(c)\sqsubseteq c'$ or $\tau(c')\sqsubseteq c$, to be small.

  It may be possible to extend our oblivious deletion construction to oblivious insertions and deletions by tweaking the parameters.  
  This was the case for bridging the gap between efficiently decoding adversarial deletions \cite{GuruswamiW17} and efficiently decoding adversarial insertions and deletions \cite{GuruswamiL16}.  
  Our analysis for deletions analyzes runs of 0s and 1s, but perhaps it would be possible to analyze insertions and deletions by relaxing the definition of a run to be ``a subword with at least 0.9 density of 0s or 1s.''  
  For now, however, we leave the case of correcting a fraction of oblivious insertions \emph{and} deletions approaching 1 as an open problem.

  \subsection{Open problems}
    These remain a number of open questions concerning deletion codes.
    Our work brings to the fore a few more.
    \begin{enumerate}
      \vspace{-1ex}
      \itemsep=0ex
  \item What is the zero-rate threshold $p_0\id{adv}$ for adversarial deletions?
  \item We showed the existence of codes that decode against $pn$ oblivious deletions for any $p<1$. Can we modify the proof to show existence of codes decoding against a combination of $pn$ oblivious insertions and deletions for every $p<1$?
  \item Can we adapt our existence proof of codes for oblivious deletions
   to obtain explicit codes that are constructable, encodable, and decodable in polynomial time?
  \item For erasures, the capacity of the random error channel and the capacity of the oblivious error channel are both $1-p$.
    For bit flips, the random and oblivious error channels also have the same capacity at $1-h(p)$.
    Can we construct codes that decode oblivious deletions with rate approaching the capacity of the random deletion channel (note this capacity itself is \emph{not} known)?
    On the other end, could we upper bound the best possible rate for correcting a fraction $p$ of oblivious deletions by $o(1-p)$ as $p\to1$?
    Such an upper bound would fundamentally distinguish the behavior of the deletions from errors and erasures.
  \item Can one find deterministic codes correcting $p$ fraction of \emph{online} deletions when $p$ approaches $\half$? By Theorem~\ref{thm:adv-on-0}, this would imply $p_0\id{adv} = 1/2$. 
  \item Given the discussion in \S\ref{sec:6-2}, can we extend our results against online deletions to stochastic codes, showing that if $p_0\id{adv}<1/2$ then there do not exist stochastic codes correcting a fraction of online deletions approaching 1/2?
\end{enumerate}

\section{Acknowledgements}
The authors would like to thank Joshua Brakensiek and Boris Bukh for helpful discussions.
The authors would also like to thank anonymous reviewers for helpful feedback on earlier versions of this paper.


\appendix


\section{Relating average-case and worst-case miscommunication probability}
\label{app:1}

In this appendix, we prove the following lemma, which shows that
constructing (deterministic) codes decodable against oblivious
deletions in the average case (over random messages, in the sense of
Theorem~\ref{thm:average-case}), suffices to construct randomized
codes with low decoding error probability for every message, as
guaranteed by Theorem~\ref{thm:oblivious}.

\begin{lemma}
  Let $p\in(0,1)$. 
  Suppose we have a family of codes $C$ of length $n$ and rate $\calR$ such that, for any deletion pattern $\tau$ with at most $pn$ deletions, at most $\eps=o_n(1)$ of the codewords are decoded incorrectly.
  Then, for any $\delta<\half$, provided $n$ is sufficiently large,  we can find a family of stochastic codes $C'$ of length $n$ and rate $\calR(1-\delta)-o(1)$ that corrects $pn$ oblivious deletions with probability $1-3\eps$.
\end{lemma}
\begin{proof}
  For set of codewords $C^*$, a pair $(c, \tau)$ (where $c\in C^*$ and $\tau$ is a deletion pattern) is said to be $C^*$-bad if $\tau(c)\subseteq c'$ for some $c'\in C^*\setminus\{c\}$, and $C^*$-good otherwise.

  Let $M = |C| = 2^{\calR n}$. For any $\tau$, let $A_\tau$ denote the set of $c$ such that $(c,\tau)$ is $C$-bad. 
  Then $|A_\tau|\le \eps |C|$ by assumption that $C$ decodes against $pn$ deletions in the average case.
  Let $d=\ceil{h(p)/\calR}$ be a constant over varying $n$, so that there are at most $M^d>\binom{n}{pn}$ choices of $\tau$.

  Let $t=M^\delta$, and $N = \floor{M^{1-\delta}/2}$. 
  We construct $C'$ iteratively.
  For $1\le k\le N$, chose a random subset of exactly $t$ codewords from $C\setminus\cup_{i=1}^{k-1}E_i$ to form $E_k$.
  With these sets of codewords, we associate each message $m_i$ with a set $E_i$.
  We encode each message $m_i$ by choosing uniformly at random one of $t$ codewords in some set $E_i$.

  Note that $C'\subseteq C$ and $|C'| = Nt = M/2$.
  It is easy to see the rate of $C'$ is $\calR(1-\delta)-o(1)$.

  We claim that, with positive probability over the choice of $C'$, $C'$ corrects $pn$ oblivious deletions with probability $1-3\eps$.
  Define $B_\tau$ to be all $c$ such that $(c,\tau)$ is $C'$-bad. 
  As $B_\tau\subseteq A_\tau$, we have $|B_\tau|\le\eps M$.
  We wish to show that the probability there exists an $B_\tau$ and a message $m_i$ with encoding set $E_i$ such that $|B_\tau\cap E_i|\ge 3\eps|E_i|$ is less than 1.
  Fix a $\tau$. 
  We show that the probability $|B_\tau\cap E_i|\ge 3\eps|E_i|$ is tiny.
  When we chose $E_i$, we can imagine we picked the $t$ elements of $E_i$ one after the other. 
  Each one was chosen from at least $M/2$ codewords, so each codeword lands in $|B_\tau\cap E_i|$ with probability at most $\eps M/(M/2) = 2\eps$.
  By Chernoff bounds, $\Pr\ba{\abs{B_\tau\cap E_i}\ge 3\eps|E_i|} \le e^{-t/3}$.
  By union bound over all $E_i$ and all $B_\tau$, we have
  \begin{equation}
    \Pr\ba{\exists E_i\exists \tau : \abs{B_\tau\cap E_i}\ge 3\eps|E_i|}
      \ \le \ N\cdot M^d\cdot e^{-t/3}
      \ < \ M^{d + 1} / e^{M^\delta/3}
  \end{equation}
  This is less than $1$ for $n > \frac{2}{\delta \calR}\log(3(d+1))$, so our code corrects $pn$ oblivious deletions with probability $1-3\eps$, as desired.
\end{proof}

\section{Comparison of oblivious deletions with oblivious bit flips}
\label{app:2}
First, for completeness, we outline the proof that that there exist codes against oblivious bit flips of rate matching the capacity of the binary symmetric channel. 
After this, we show that it is not possible to follow the same approach to construct codes against oblivious deletions.

\begin{theorem}[\cite{Langberg08,CsiszarN88}]
  For every $p \in [0,1/2)$, and every $\calR < 1-h(p)$, there exists a binary code family of rate $\calR$ that can correct a fraction $p$ of oblivious bit flips.
 \footnote{Note that $1-h(p)$ is an upper bound on the rate as the capacity of the binary symmetric channel with crossover probability $p$ is $1-h(p)$ and the oblivious channel can always just choose the error vector randomly.}
  \label{thm:app-2-1}
\end{theorem}
\begin{proof}
We consider a stochastic code $C$ where each message $m$ is mapped uniformly into a set $E(m)$ of size $t$. 
  Choose $t=n$ and $\eps = 1/\log n$.
  This code decodes successfully if, for all $m$,
  \begin{equation}
    \Pr_{c\in E(m)}\left[ c+e\in \bigcup_{c\in C\setminus E(m)}B_{pn}(c')\right]
      \ \ge \ 1-\eps
  \end{equation}
  where $B_{pn}(x)\subseteq\{0,1\}^n$ is the set of all words within Hamming distance $pn$ of $x$.

  Choose $C$ to have $2^{\calR n}$ codewords at random.
  Fix a message $m$, the rest of the codeword $C\setminus E(m)$, and the error vector $e$. We have
  \begin{equation}
    \Pr_{c\sim\{0,1\}^n}\left[c\in \bigcup_{c\in C\setminus E(m)}B_{pn}(c'+e)\right]
      \ \le \ \frac{1}{2^n}|C|\cdot|B_{pn}(0)|
      \ \approx \ \frac{2^{\calR n}\cdot 2^{h(p)n}}{2^n},
  \end{equation}
  which is $2^{-\Omega(n)}$ if $\calR < 1-h(p)$. 
  In this case 
  \begin{align}
    \Pr_{E(m)}&\ba{
      \Pr_{c\in E(m)}\ba{c+e\notin \bigcup_{c\in C\setminus E(m)}B_{pn}(c')}
        < 1-\eps
    } \nonumber\\
      &= \Pr_{E(m)}\ba{
        \Pr_{c\in E(m)}\ba{c\notin \bigcup_{c\in C\setminus E(m)}B_{pn}(c'+e)}
          < 1-\eps
      } \nonumber\\
      &= \Pr_{E(m)}\ba{
        \abs{E(m)\cap \bigcup_{c\in C\setminus E(m)}B_{pn}(c'+e)}
          > \eps t
      } \nonumber\\
      &\le 2^{-\Omega(\eps nt)} = 2^{-\tilde\Omega(n^2)}.
  \end{align}
  Thus, applying a union bound over all $m$ and error vectors $e$, we get
  \begin{align}
    \Pr_C\ba{\exists m:
      \Pr_{c\in E(m)}\ba{c+e\in \bigcup_{c\in C\setminus E(m)}B_{pn}(c')}
        > \eps
    } 
    \le 2^{\calR n}\cdot 2^n\cdot 2^{-\tilde\Omega(n^2)} \le 2^{-\tilde\Omega(n^2)}.
  \end{align}
  Thus when $\calR<1-h(p)$, we can construct a rate $\calR$ stochastic code correcting w.h.p. $p$ oblivious bit flips.
\end{proof}

For oblivious deletions the above technique does not work.
This is established by the following lemma.
\begin{lemma}
  \label{lem:app-2-2}
  Let $C$ be a random length $n$ rate $\calR$ stochastic binary code such that each message is encoded uniformly at random in one of $t$ codewords.
  If $p\ge0.4$ and $0<\eps<1$, then for any message $m$ and deletion pattern $\tau$ we have
  \begin{equation}
    \Pr_{C}\ba{\Pr_{c\in E(m)}\ba{\exists c'\in C\setminus E(m):\tau(c)\sqsubseteq c'} > \eps } \ > \ 2^{-h(p)n}
  \end{equation}
  where $D_{k}(s)=\{s':|s'|=|s|-k, s'\le s\}$.
\end{lemma}
In plain English, this lemma says that for a fixed message and deletion pattern, the probability that message decodes incorrectly too many times ($>\eps$ fraction of the time) is too large ($>2^{-h(p)n}$). In short, the reason the lemma is true is because the probability the entire code $C$ contains a ``really bad'' word (e.g. a word with at least $0.91n$ runs of 1s) is too large.
The following remark shows us, using Lemma~\ref{lem:app-2-2}, why we cannot follow the same randomized approach as Theorem~\ref{thm:app-2-1}.

\begin{remark}
  Following the approach of Theorem~\ref{thm:app-2-1}, we would like to conclude, using union bound,
  \begin{align}
    \Pr_{C}&\ba{\forall m\forall \tau\,\Pr_{c\in E(m)}\ba{\exists c'\in C\setminus E(m):\tau(c)\sqsubseteq c'} > \eps }  \nonumber\\
    &\le\sum_m\sum_\tau\Pr_{C}\ba{\Pr_{c\in E(m)}\ba{\exists c'\in C\setminus E(m):\tau(c)\sqsubseteq c'} > \eps } \nonumber\\ 
    &\le 2^{\calR n}\binom{n}{pn}\Pr_{C}\ba{\Pr_{c\in E(m)}\ba{\exists c'\in C\setminus E(m):\tau(c)\sqsubseteq c'} > \eps } \nonumber\\ 
    &\red{<2^{\calR n}2^{h(p)n}2^{-h(p)n-\Omega(n)}} = 2^{-\Omega(n)}
  \end{align}
  Unfortunately the last inequality (indicated in red) is false by Lemma \ref{lem:app-2-2}.
\end{remark}
\begin{proof}[Proof of Lemma~\ref{lem:app-2-2}]
  It suffices to prove for $p=0.4$. Indeed, as $p$ increases, 
  \begin{equation}
    \Pr_{c\in E(m)}\ba{\tau(c)\in \bigcup_{c'\in C\setminus E(m)}D_{pn}(c')}
  \end{equation}
  increases as $\tau(c)$ contains fewer symbols.

  If $p=0.4$, then $\tau(c)$ has $0.6n$ characters. Note that if $c$ is uniform on $\{0,1\}^n$, then $\tau(c)$ is uniform on $\{0,1\}^{0.6n}$.
  It is easy to check that for a uniformly chosen $0.6n$ string, the probability it is a subsequence of the length $0.91n$ string $a_{0.91n} = 0101\dots01$ is $1-2^{-\Omega(n)}$: the position in the longest string increases by 1.5 each time, so in expectation the string spans $0.9n$ characters of $a_{0.91n}$, so it is not a subsequence with exponentially small probability by Chernoff bounds.

  There are $2^{0.09n}$ length n strings that start with $a_{0.91n}$. Call this set of strings $A$.
  Thus, the probability that $C$ contains an element of $A$ is at least $2^{-0.91n}$ since that is the probability the first element is an element of $A$.
  Conditioned on $C$ containing an element of $A$, the probability that, for some $c\in E(m)$, $\tau(c)$  is a subsequence of some element of $A$ is at least the probability that $c$ is a subsequence of $a_{0.91n}$, which is $1-2^{-\Omega(n)}$. 
  Thus, conditioned on $C$ containing an element of $A$, we have
  \begin{equation}
    \E\ba{\#\{c:\exists c'\in C\setminus E(m)\suchthat \tau(c)\le c'\} } \ = \ t(1-2^{-\Omega(n)})
  \end{equation}
  Conditioned on $C$ being fixed, consider the indicator random variables $X_i$ for whether each $c_i\in E(m)$ satisfies $\tau(c_i)$ is a subsequence of some $c'\in C\setminus E(m)$. The $X_1,\dots,X_t$ are i.i.d, so the probability $\sum {X_i} / t >\eps$ is approximately 1, (for sure, it is $1-2^{-\Omega(n)}$).

  Thus we conclude
  \begin{equation}
    \Pr_C[
      \#\{c:\exists c'\in C\setminus E(m)\suchthat \tau(c)\le c'\} > \eps t
    ] \ge (1-2^{-\Omega(n)})\Pr_C[C\cap A\nonempty] > 2^{(0.91-o(1))n} > 2^{-H(0.4)n}
  \end{equation}
  as $H(0.4)\approx 0.97$.
\end{proof}

Intuitively \ref{lem:app-2-2} should be true as deletion codes behave poorly when chosen completely randomly. In the adversarial deletion channel, random codes correct only $0.17n$ deletions \cite{KashMTU11}, while the best known constructions correct $0.41n$ deletions \cite{BukhGH17}.

\section{Proofs of Lemma~\ref{lem:geom-1} and Lemma~\ref{lem:geom-2}}
\label{app:3}
	\begin{lemma}
		\label{lem:geom-app}
		For any $j\in[1,K]$, if $Z$ is a random variable distributed as $\min(\Geometric(j/K),\sqrt{R}-1)$, then $\E[Z]>\frac{K}{2j}-1$.
	\end{lemma}
  \begin{proof}
		We have
		\begin{align}
			\E[D]
				\ &= \ \left( \frac{j}{K} \right)\cdot 1 + \left( \frac{j}{K} \right)\left( 1-\frac{j}{K} \right)\cdot 2 + \cdots + \left( \frac{j}{K} \right)\left( 1-\frac{j}{K} \right)^{\sqrt{R}-2}\cdot(\sqrt{R}-1) \nonumber\\
				\ &= \ \frac{1-\left( 1-\frac{j}{K} \right)^{\sqrt{R}-1}}{j/K} - (\sqrt{R}-1)\left( 1-\frac{j}{K} \right)^{\sqrt R - 1} \nonumber \\ 
				\ &\ge \ \frac{1-\left( 1-\frac{1}{K} \right)^{\sqrt{R}-1}}{j/K} - (\sqrt{R}-1)\left( 1-\frac{1}{K} \right)^{\sqrt R - 1} 
				\ > \ \frac{K}{2j} - 1
		\end{align}
		The last inequality follows from recalling $R=4K^4$ and twice applying
		\begin{equation}
			\left( \sqrt R - 1 \right)\left( \frac{K-1}{K} \right)^{\sqrt R - 1}
			\ < \ 2K^2\cdot\left( \frac{K-1}{K} \right)^{K^2}
			\ < \ 2K^2\cdot\left( \half \right)^K < 1,
		\end{equation}
		which is true since $K>8$.
  \end{proof}

  \begin{lemma*}[Lemma~\ref{lem:geom-1}]
    Let $J$ be chosen uniformly from $[K]$.
    Let $D$ be a random variable that is 1 if $J\in[\lambda-1]$ and, conditioned on a fixed $J\ge \lambda$, is distributed as $\min(\Geometric(J/K),\sqrt R)$.
    Then $\E[D]\ge (\log K)/4$.
  \end{lemma*}
  \begin{proof}
    Applying Lemma~\ref{lem:geom-app},
		\begin{align}
			\E[D]
      \ &= \ \frac{1}{K}\sum_{j=1}^{K} \E[D|J=j] 
			\ > \ \sum_{j=1}^{\lambda-1}\frac{1}{K}\cdot 1 + \sum_{j=\lambda}^{K}\frac{1}{K}\left( \frac{K}{2j} - 1 \right) \nonumber\\
			\ &\ge \ -1 + \sum_{j=\lambda}^{K}\left( \frac{1}{2j} \right) 
			\ > \ \half(\ln K - \ln\lambda - 1) - 1  
			\ > \ \fourth \log K.
		\end{align} 
  \end{proof}
  \begin{lemma*}[Lemma~\ref{lem:geom-2}]
    Let $\lambda'\in[\lambda,K]$ and let $J$ be chosen uniformly from $\{\lambda,\lambda+1,\dots,\lambda'\}.$
    Let $D$ be the random variable that, conditioned on a fixed $J$, is distributed as $\min(\Geometric(J/K),\sqrt R)$.
    Then $\E[D]\ge (\log K)/4$.
  \end{lemma*}
  \begin{proof}
    Applying Lemma~\ref{lem:geom-app},
		\begin{align}
			\E[D]
      \ &= \ \frac{1}{\lambda'-\lambda+1}\sum_{j=\lambda}^{\lambda'} \E[D|J=j] 
			\ > \ \frac{1}{\lambda'-\lambda+1}\sum_{j=\lambda}^{\lambda'}\left( \frac{K}{2j}-1 \right) \nonumber\\
			\ &> \ \frac{1}{K-\lambda+1}\sum_{j=\lambda}^{K}\left( \frac{K}{2j}-1 \right) 
			\ > \ -1 + \sum_{j=\lambda}^{K}\left( \frac{1}{2j} \right) 
			\ > \ \fourth\log K. \qedhere
			\label{eq:prop-5}
		\end{align}
  \end{proof}

\end{document}